\documentclass[english]{article}
\usepackage[T1]{fontenc}
\usepackage[utf8]{inputenc}
\usepackage{color}
\usepackage{colortbl}
\usepackage{babel}
\usepackage{array}
\usepackage{float}
\usepackage{booktabs}
\usepackage{mathtools}
\usepackage{multirow}
\usepackage{amsmath}
\usepackage{amsthm}
\usepackage{amssymb}
\usepackage{geometry}
\usepackage{rotating}
\usepackage[pdfusetitle,
 bookmarks=true,bookmarksnumbered=false,bookmarksopen=false,
 breaklinks=false,pdfborder={0 0 0},pdfborderstyle={},backref=false,colorlinks=true]
 {hyperref}

\makeatletter
\providecommand{\tabularnewline}{\\}
\floatstyle{ruled}
\newfloat{algorithm}{tbp}{loa}
\providecommand{\algorithmname}{Algorithm}
\floatname{algorithm}{\protect\algorithmname}

\usepackage[utf8]{inputenc} 
\usepackage[T1]{fontenc}    
\usepackage{url}            
\usepackage{booktabs}       
\usepackage{amsfonts}       
\usepackage{nicefrac}       
\usepackage{microtype}      
\usepackage{xcolor}         
\usepackage{algorithmic} 

\makeatother

\theoremstyle{plain}
\newtheorem{thm}{\protect\theoremname}[section]
\newtheorem{cor}{\protect\corollaryname}[section]
\theoremstyle{definition}
\newtheorem{defn}{\protect\definitionname}[section]
\theoremstyle{plain}
\newtheorem{lem}{\protect\lemmaname}[section]
\theoremstyle{remark}
\newtheorem{rem}{\protect\remarkname}[section]
\providecommand{\corollaryname}{Corollary}
\providecommand{\definitionname}{Definition}
\providecommand{\lemmaname}{Lemma}
\providecommand{\remarkname}{Remark}
\providecommand{\theoremname}{Theorem}
\begin{document}
\global\long\def\threshold{\mathrm{Threshold}}%
\global\long\def\Lap{\mathrm{Lap}}%
\global\long\def\N{{\cal N}}%
\global\long\def\R{\mathbb{R}}%
\global\long\def\diag{\mathrm{diag}}%
\global\long\def\cov{\mathrm{Cov}}%
\global\long\def\E{\mathbb{E}}%
\global\long\def\poly{\mathrm{poly}}%
\global\long\def\sign{\mathrm{sign}}%
\global\long\def\Head{\mathrm{Head}}%
\global\long\def\Tail{\mathrm{Tail}}%
\global\long\def\AT{\textsc{AboveThreshold}}%
\global\long\def\epstot{\epsilon_{\mathrm{tot}}}%
\global\long\def\deltot{\delta_{\mathrm{tot}}}%
\global\long\def\epsit{\epsilon_{\mathrm{it}}}%
\global\long\def\deltit{\delta_{\mathrm{it}}}%

\allowdisplaybreaks
\title{Adaptive Power Iteration Method for Differentially Private PCA}
\author{Ta Duy Nguyen\thanks{\texttt{taduy@bu.edu}.} \and
Alina Ene\thanks{Department of Computer Science, Boston University, \texttt{aene@bu.edu}.} \and
Huy L. Nguyen\thanks{Khoury College of Computer and Information Science, Northeastern University, \texttt{hu.nguyen@northeastern.edu}.}}
\date{}
\maketitle
\begin{abstract}

We study $\left(\epsilon,\delta\right)$-differentially private algorithms
for the problem of approximately computing the top singular vector
of a matrix $A\in\R^{n\times d}$ where each row of $A$ is a data
point in $\R^{d}$. Following Dwork-Talwar-Thakurta-Zhang (STOC 2014),
we consider the privacy model where neighboring inputs differ by one
single row. We give a novel algorithm that achieves beyond-worst-case
guarantees for input matrices with low coherence, which is a structural
property of matrices in many applications, including but not limited
to i.i.d. data. Our algorithm contributes to the extensive literature
on private power iteration methods, where we introduce a new filtering
technique which adapts to this coherence parameter. Our work departs
from and complements the work by Hardt-Roth (STOC 2013) which achieves
beyond-worst-case guarantees for the more restrictive privacy model
where neighboring inputs differ in one single entry by at most 1.
\end{abstract}

\section{Introduction}

In this paper, we study $\left(\epsilon,\delta\right)$-differentially
private algorithms for the problem of approximately computing the
top singular vector of a matrix $A\in\R^{n\times d}$, where each
row of $A$ is a datapoint in $\R^{d}$ that constitutes the private
data of an individual, and each column is an attribute. The problem
of computing the top singular vector is a fundamental primitive with
many applications in theoretical computer science, machine learning,
statistics, data mining, and many other areas. In particular, it is
the main component of algorithms for widely used algorithmic primitives
such as principal component analysis (PCA) and computing low-rank
factorizations and approximations of matrices \cite{blum2020foundations}.
Importantly, in many applications, the input matrix records sensitive
data and providing strong privacy guarantees is a non-negotiable constraint
\cite{DBLP:conf/kdd/McSherryM09,DBLP:conf/sp/NarayananS08}. This
important consideration has motivated the development of algorithms
that preserve the privacy of the individuals while simultaneously
constructing accurate solutions with the best possible approximation
guarantees.

The design of private algorithms for top singular vector computation
and related matrix computations has received significant attention,
leading to the development of algorithms that provide several different
levels of privacy protections \cite{DBLP:conf/stoc/DworkTT014,DBLP:conf/stoc/HardtR13,DBLP:journals/jmlr/ChaudhuriSS13,DBLP:conf/stoc/HardtR12,DBLP:conf/pods/BlumDMN05,DBLP:conf/kdd/McSherryM09,DBLP:journals/jpc/DworkMNS16,DBLP:conf/soda/KapralovT13,DBLP:conf/nips/HardtP14,DBLP:conf/colt/BalcanDWY16,DBLP:conf/nips/MangoubiV22a,DBLP:conf/nips/LiuK0O22}.
In particular, this line of work led to a complete answer to the question
of what is the best additive approximation guarantee that is achievable
in the worst case. The work of \cite{DBLP:conf/stoc/DworkTT014} assumes
that each row has $\ell_{2}$ norm bounded by 1. Under this constraint,
they give an algorithm that achieves an additive approximation of
$\widetilde{O}(\sqrt{d})$ and prove that this guarantee is best possible
up to logarithmic factors. This norm-bounded assumption is well-suited
for PCA: without it, a single row with a large norm could dominate
the PCA solution, undermining its core purpose of capturing the underlying
structure shared by most of the data. Because PCA is inherently a
least-squares method minimizing the sum of squared residuals, using
$\ell_{2}$ for the norm bound is a natural and appropriate choice,
which we follow.

Notably, these results show that there is a significant cost to ensuring
privacy for worst-case data, which may make the algorithms completely
unusable in practice, such as in high-dimensional data analysis applications
\cite{DBLP:conf/stoc/HardtR13}. Fortunately, matrices arising in
applications are far from worst-case instances, which motivates the
development of algorithms that leverage the structure to achieve improved
approximation guarantees while still ensuring the same privacy guarantees.
In their seminal work, \cite{DBLP:conf/stoc/HardtR13} made significant
progress in this direction \textit{in a different model of privacy}.
Inspired by the influential line of work on matrix completion initiated
by \cite{DBLP:journals/tit/CandesT10}, Hardt and Roth \cite{DBLP:conf/stoc/HardtR13}
investigated whether stronger utility guarantees can be obtained for
matrices with low coherence, where the coherence value is proportional
to the maximum entry of the singular vectors of the matrix. This question
fits naturally within the broader research program on compressive
sensing, which moves beyond the classical assumption of i.i.d. data
and instead exploits structural properties of the data such as coherence
to develop a richer statistical theory.  The coherence in particular
is a compelling parameter to work with: it tends to be small in a
wide range of practically relevant settings, including both random
data and structured non-random data arising in applications.

Hardt and Roth \cite{DBLP:conf/stoc/HardtR13} design an algorithm
that provides privacy protection \textit{for each individual entry
of the matrix, }and whose error guarantee scales with the coherence
rather than the ambient dimension, yielding significantly stronger
bounds when the matrix is well-structured. However, it can only protect
each \textit{individual entry} in the matrix. Subsequent works including
\cite{DBLP:conf/nips/HardtP14,NicolasSMMC25,dOrsiN26} attempt to
generalize this privacy model but when an \emph{entire row} of the
matrix changes, they only provide guarantees when the $\ell_{1}$
norm of the change is bounded. For this reason, these results cannot
be applied directly in the natural model which aims to protect against
any individual row change measured in the $\ell_{2}$ norm. Despite
the sustained interest from the community over the last decade \cite{DBLP:conf/pods/BlumDMN05,DBLP:journals/jpc/DworkMNS16,DBLP:conf/soda/KapralovT13,DBLP:journals/jmlr/ChaudhuriSS13,DBLP:conf/stoc/DworkTT014,DBLP:conf/nips/MangoubiV22a,DBLP:conf/nips/LiuK0O22},
obtaining beyond-worst-case guarantees for this more general model
has remained an open question. Motivated by these considerations,
our work addresses the following question:
\begin{center}
\textit{In the privacy model where each row of the matrix is the private
data of a user, is it possible to design an algorithm that is private
on all input matrices and achieves beyond-worst-case guarantees for
input matrices with low coherence?}
\par\end{center}

In this work, we give an affirmative answer to this question. Our
algorithm does not require any knowledge of the coherence parameter
and it automatically adapts to it. In contrast to prior works, our
algorithm ensures privacy for each row via a novel adaptive filtering
technique and it achieves fine-grained utility guarantees via a new
analysis. In addition, our algorithm is a private variant of the classical
power iteration method that is widely used in practice.

While the main focus of our work is on deterministic data with beyond-worst-case
guarantees, our algorithm can be applied to random data sampled i.i.d
from a distribution. In particular, when we have i.i.d data from a
Gaussian distribution, which naturally has a low coherence parameter,
our algorithm provides guarantees that compare favorably with the
state-of-the-art results. We give an overview of our contributions
and techniques in the following subsections.

Our algorithm contributes to the broader line of work that developed
private algorithms with fine-grained and beyond-worst-case guarantees
for other fundamental optimization problems, such as private mean
estimation \cite{DBLP:conf/nips/BrownGSUZ21,DBLP:conf/colt/BrownHS23,DBLP:conf/colt/KuditipudiDH23,DBLP:journals/corr/abs-2502-18698}
and regression \cite{DBLP:conf/nips/DickGJ23,DBLP:conf/colt/0003HHKLOPS24}.

\subsection{Our contributions}

Let $A\in\R^{n\times d}$ be the input matrix. As in \cite{DBLP:conf/stoc/DworkTT014},
we assume that every row $a$ of $A$ has length at most $1$, i.e
$\left\Vert a\right\Vert _{2}\le1$; this condition is only necessary
in order to achieve non-trivial utility guarantees \cite{DBLP:conf/stoc/DworkTT014}.
Let $A=U\Sigma V^{\top}$ be the compact SVD decomposition of $A$
where $U\in\R^{n\times d}$, $\Sigma=\diag(\sigma_{1},\dots,\sigma_{d})\in\R^{d\times d}$
is the diagonal matrix of singular values of $A$ and $V\in\R^{d\times d}$
is a unitary matrix. Note that $A^{\top}A=V\Sigma^{2}V^{\top}$. We
assume that $\sigma_{1}\geq\sigma_{2}\geq\dots\geq\sigma_{d}$. We
let $\left\Vert U\right\Vert _{\infty}=\max_{i,j}\left|U_{ij}\right|$
and $\Upsilon=\max_{i}\left|U_{i,1}\right|$. We have $\Upsilon\le\left\Vert U\right\Vert _{\infty}\leq\sqrt{\frac{\mu\left(A\right)}{n}}$,
where $\mu\left(A\right)\coloneqq\max\left\{ n\left\Vert U\right\Vert ^{2}_{\infty},d\left\Vert V\right\Vert ^{2}_{\infty}\right\} $
is the coherence parameter (defined in \cite{DBLP:conf/stoc/HardtR13}).
Because every row of $A$ has length at most 1, we have $\sigma_{1}\Upsilon\le1$.
For two unit-length vectors $u$ and $v$, we use $\sin(u,v)$ to
denote the sine of the angle between $u$ and $v$, i.e., $\sin(u,v)=\sin(\arccos(\langle u,v\rangle))$.

The main contribution of our paper is a private power iteration algorithm,
presented in Algorithm \ref{alg:Algorithm} with the following guarantee.
\begin{thm}
\label{thm:main}Let $\epsilon,\delta,\beta\in(0,1]$. Let $\kappa=\frac{\sigma^{2}_{1}-\sigma^{2}_{2}}{\sigma^{2}_{1}}$
and let $v_{1}$ be the top singular vector of $A$. Suppose we run
Algorithm \ref{alg:Algorithm} with per-iteration privacy budget $\epsilon_{0},\delta_{0}$,
such that $\epsilon=2\epsilon_{0}\sqrt{2\ln\left(1/\delta_{0}\right)T}+2\epsilon_{0}\frac{e^{2\epsilon_{0}}-1}{e^{2\epsilon_{0}}+1}T$
and $\delta=\delta_{0}\left(T+1\right)$. Then the algorithm is $(\epsilon,\delta)$
differentially private (Definition \ref{privacy-defn}). Furthermore,
if the input matrix $A\in\R^{n\times d}$ satisfies $n\geq d$, each
row of $A$ has length at most $1$, and $\kappa\geq\frac{16\sqrt{2\ln\left(\left(T+1\right)/\delta\right)T}}{\epsilon}\left(\frac{K}{\sigma_{1}}\Upsilon+\frac{K^{2}}{\sigma^{2}_{1}}\right)$
 where $K=O\left(\log\frac{Tn}{\beta\delta}\right)$, then with probability
at least $\frac{3}{4}-\beta$, Algorithm \ref{alg:Algorithm} outputs
a vector $x=\frac{x^{(T)}}{\left\Vert x^{(T)}\right\Vert _{2}}$ such
that $\sin^{2}\left(v_{1},x\right)\leq B$ and $\left\Vert A^{\top}Ax\right\Vert _{2}\ge\frac{\sigma^{2}_{1}}{1+B},$
where $B=O\Big(\Big(\frac{\sqrt{\min\left\{ 4n/\sigma^{2}_{1},d\right\} }}{\sqrt{\kappa}}+\frac{1}{\kappa}+\sqrt{d}\Big)\frac{K\Upsilon\sqrt{\ln\left(T/\delta\right)T}}{\epsilon\sigma_{1}}+KT\sqrt{d}\left(1+\frac{\kappa}{2}\right)^{-T}\Big)^{2}$
.
\end{thm}
In the following corollary, we show the guarantee achieved by our
algorithm when $T=\widetilde{\Theta}\left(\frac{1}{\kappa}\right)$.
To attain a comparable guarantee (up to a logarithmic factor) while
maintaining privacy, we can use the following standard approach: guess
$\kappa$ up to a constant factor, run the algorithm for each guess
separately with $T$ set based on the guessed value, and return the
approximately best solution found using the exponential mechanism.
\begin{cor}
\label{cor:main}Let $\beta>0$. Under the same conditions as in Theorem
\ref{thm:main}, if the number of iterations is $T=\Theta\left(\frac{\log\frac{n}{\beta\delta\epsilon}}{\kappa}\right)$,
with probability $\ge\frac{3}{4}-\beta$, Algorithm \ref{alg:Algorithm}
outputs a vector $x=\frac{x^{(T)}}{\left\Vert x^{(T)}\right\Vert _{2}}$
such that 
\begin{align*}
\sin^{2}\left(v_{1},x\right) & =\widetilde{O}\left(\left(\frac{\min\left\{ 4n/\sigma^{2}_{1},d\right\} }{\kappa^{2}}+\frac{1}{\kappa^{3}}+\frac{d}{\kappa}\right)\frac{1}{\epsilon^{2}\sigma^{2}_{1}}\Upsilon^{2}\right),
\end{align*}
where $\widetilde{O}\left(\cdot\right)$ is suppressing $\log d$,
$\log n$, $\log(1/\epsilon)$, $\log(1/\delta)$, $\log(1/\beta)$
factors.

\end{cor}
We also consider the stochastic case where the rows of $A$ are i.i.d.
samples from a Gaussian distribution $\N(0,\overline{\Sigma}^{2})$
where $\overline{\Sigma}^{2}\in\R^{d\times d}$ is a PSD matrix. In
this case, the following theorem provides the error guarantee of Algorithm
\ref{alg:Algorithm}, measured with respect to the top eigenvector
of $\overline{\Sigma}^{2}$.
\begin{thm}
\label{thm:gaussian}Let $\beta>0$. Suppose the rows of $A$ are
i.i.d. samples from a Gaussian distribution $\N(0,\overline{\Sigma}^{2})$
where $\overline{\Sigma}^{2}\in\R^{d\times d}$ is a PSD matrix. Let
$\overline{\sigma}^{2}_{1}\ge\dots\ge\overline{\sigma}^{2}_{d}$ be
the eigenvalues of $\overline{\Sigma}^{2}$, $\overline{v}_{1}$ be
the top eigenvector of $\overline{\Sigma}^{2}$, and $\overline{\kappa}=\frac{\overline{\sigma}^{2}_{1}-\overline{\sigma}^{2}_{2}}{\overline{\sigma}^{2}_{1}}$.
{Let $\epsilon,\delta\in(0,1)$. Run Algorithm \ref{alg:Algorithm} with per-iteration budgets $\epsilon_0,\delta_0$ calibrated to total budget $(\epsilon,\delta)$ as in Theorem \ref{thm:main}; the algorithm is then $(\epsilon,\delta)$ differentially private. Assume $\overline{\sigma}_1>0$ and $\overline{\kappa}>0$.} Suppose that $\overline{\sigma}^{2}_{1}+\dots+\overline{\sigma}^{2}_{d}=1$,
$T=\Theta\left(\frac{1}{\overline{\kappa}}\log\frac{n}{\beta\delta\epsilon}\right)$
and $L=\Theta\left(\sqrt{\log\frac{n}{\beta}}\right)$ where $\Theta\left(\cdot\right)$
is suppressing sufficiently large absolute constants. If $n\geq{\widetilde{\Theta}\left(d+\frac{1}{\overline{\sigma}^{2}_{1}\overline{\kappa}^{2}}+\frac{1}{\epsilon\overline{\sigma}^{2}_{1}\overline{\kappa}^{3/2}}\right)}$
sufficiently large, with probability at least $\frac{3}{4}-\beta$,
when run on input $\frac{1}{L}A$, Algorithm \ref{alg:Algorithm}
outputs a vector $x=\frac{x^{(T)}}{\left\Vert x^{(T)}\right\Vert _{2}}$
such that
\begin{align*}
\sin^{2}\left(\overline{v}_{1},x\right) & \le\widetilde{O}\left(\frac{d}{\epsilon^{2}n^{2}\overline{\sigma}^{2}_{1}\overline{\kappa}^{2}}+\frac{1}{\epsilon^{2}n^{2}\overline{\sigma}^{2}_{1}\overline{\kappa}^{3}}+\frac{1}{n\overline{\sigma}^{2}_{1}\overline{\kappa}^{2}}+\frac{1}{n^{2}\overline{\sigma}^{4}_{1}\overline{\kappa}^{2}}\right),
\end{align*}

where $\widetilde{O}\left(\cdot\right)$ is suppressing $\log d$,
$\log n$, ${\log(1/\epsilon)}$, ${\log(1/\delta)}$, $\log(1/\beta)$
factors.
\end{thm}
\textbf{Comparison with prior work}. We summarize our results in comparison
with prior work in Table \ref{tab:result}. 

\begin{table*}
\caption{Summary of results. In the deterministic case, $\kappa=\frac{\sigma^{2}_{1}-\sigma^{2}_{2}}{\sigma^{2}_{1}}$
denotes the eigengap of $A^{\top}A$, and $\Upsilon=\max_{i}\left|U_{i,1}\right|\protect\leq\sqrt{\frac{\mu\left(A\right)}{n}}$
where $\mu(A)$ is the coherence parameter. In the Gaussian i.i.d.
data case, $\overline{\kappa}=\frac{\overline{\sigma}^{2}_{1}-\overline{\sigma}^{2}_{2}}{\overline{\sigma}^{2}_{1}}$
denotes the eigengap of the covariance matrix $\overline{\Sigma}^{2}$
of the distribution. $\widetilde{O}$ and $\widetilde{\Omega}$ notations
suppress $\log d$, $\log n$, $\log(1/\epsilon)$, $\log(1/\delta)$,
$\log(1/\beta)$ factors.}
\label{tab:result}
\begin{centering}
\begin{tabular*}{1\textwidth}{@{\extracolsep{\fill}}>{\raggedright}m{0.007\textwidth}>{\raggedright}m{0.15\textwidth}>{\raggedright}m{0.43\textwidth}>{\raggedleft}m{0.3\textwidth}}
\toprule 
 & Algorithm & $\sin^{2}$-Error & Condition\tabularnewline
\midrule
\midrule 
\multirow{2}{0.007\textwidth}[-0.8cm]{\begin{turn}{90}
\raggedright{}Deterministic data
\end{turn}} & \cite{DBLP:conf/stoc/HardtR13,DBLP:conf/nips/HardtP14} (Entry private\footnotemark[1]) & $\widetilde{O}\left(\frac{n\left\Vert U\right\Vert ^{2}_{\infty}\sigma_{1}}{\epsilon^{2}\left(\sigma_{1}-\sigma_{2}\right)^{3}}\right)$ & $\sigma_{1}-\sigma_{2}=\widetilde{\Omega}\left(\frac{\left\Vert U\right\Vert _{\infty}}{\epsilon}\right)$\tabularnewline
\cmidrule{2-4}
 & \cite{DBLP:conf/stoc/HardtR13,DBLP:conf/nips/HardtP14} + Composition & $\widetilde{O}\left(\frac{nd\left\Vert U\right\Vert ^{2}_{\infty}\sigma_{1}}{\epsilon^{2}\left(\sigma_{1}-\sigma_{2}\right)^{3}}\right)$ & $\sigma_{1}-\sigma_{2}=\widetilde{\Omega}\left(\frac{\left\Vert U\right\Vert _{\infty}\sqrt{d}}{\epsilon}\right)$\tabularnewline
\cmidrule{2-4}
 & \cite{dOrsiN26}\footnotemark[2] & $\widetilde{O}\left(\frac{nd\left\Vert U\right\Vert ^{2}_{\infty}}{\epsilon^{2}\left(\sigma_{1}-\sigma_{2}\right)^{2}}\right)$ & $\sigma_{1}-\sigma_{2}=\widetilde{\Omega}\left(\frac{\left\Vert U\right\Vert _{\infty}\sqrt{d}}{\epsilon}\right)$\tabularnewline
\cmidrule{2-4}
 & \cite{DBLP:conf/stoc/DworkTT014} & $\widetilde{O}\left(\frac{d}{\epsilon^{2}\left(\sigma^{2}_{1}-\sigma^{2}_{2}\right)^{2}}\right)$ & $\sigma^{2}_{1}-\sigma^{2}_{2}=\widetilde{\Omega}\left(d^{1/2}\right)$\tabularnewline
\cmidrule{2-4}
 & Our algorithm (Cor. \ref{cor:main}) & $\widetilde{O}\left(\frac{\min\left\{ 4n,d\sigma^{2}_{1}\right\} \Upsilon^{2}+\sigma^{2}_{1}\Upsilon^{2}/\kappa+d\sigma^{2}_{1}\Upsilon^{2}\kappa}{\epsilon^{2}\left(\sigma^{2}_{1}-\sigma^{2}_{2}\right)^{2}}\right)$ & $\sigma^{2}_{1}-\sigma^{2}_{2}={\widetilde{\Omega}\left(\frac{\Upsilon\sigma_1+1}{\epsilon\sqrt{\kappa}}\right)}$\tabularnewline
\midrule 
\multirow{3}{0.007\textwidth}[0.3cm]{\begin{turn}{90}
Gaussian i.i.d. data
\end{turn}} & \cite{DBLP:conf/nips/LiuK0O22}

Algorithm 2 & $\widetilde{O}\left(\frac{d}{n\overline{\kappa}^{2}}+\frac{d^{2}}{\epsilon^{2}n^{2}\overline{\kappa}^{2}}\right)$ & $\epsilon=\widetilde{O}\left(\frac{1}{\sqrt{n}}\right)$\tabularnewline
\cmidrule{2-4}
 & \cite{DBLP:conf/nips/LiuK0O22}

Algorithm 3 & $\widetilde{O}\left(\frac{d}{n\overline{\kappa}^{2}}+\frac{d^{2}}{\epsilon^{2}n^{2}\overline{\kappa}^{2}}\right)$ & $n\geq\widetilde{\Theta}\left(e^{1/\overline{\kappa}^{2}}+\frac{1}{\overline{\sigma}^{2}_{1}\overline{\kappa}^{2}}+\frac{d}{\epsilon\overline{\kappa}}\right)$\tabularnewline
\cmidrule{2-4}
 & Our algorithm (Thm. \ref{thm:gaussian}) & $\widetilde{O}\left(\frac{d}{\epsilon^{2}n^{2}\overline{\sigma}^{2}_{1}\overline{\kappa}^{2}}+\frac{1}{\epsilon^{2}n^{2}\overline{\sigma}^{2}_{1}\overline{\kappa}^{3}}+\frac{1}{n\overline{\sigma}^{2}_{1}\overline{\kappa}^{2}}+\frac{1}{n^{2}\overline{\sigma}^{4}_{1}\overline{\kappa}^{2}}\right)$ & $n\geq{\widetilde{\Theta}\left(d+\frac{1}{\overline{\sigma}^{2}_{1}\overline{\kappa}^{2}}+\frac{1}{\epsilon\overline{\sigma}^{2}_{1}\overline{\kappa}^{3/2}}\right)}$\tabularnewline
\bottomrule
\end{tabular*}
\par\end{centering}
\end{table*}

\textbf{Deterministic case}. The best known guarantees for the problem
are from the works of \cite{DBLP:conf/stoc/DworkTT014}, \cite{DBLP:conf/stoc/HardtR13,DBLP:conf/nips/HardtP14}
and \cite{dOrsiN26} (shown in Table \ref{tab:result}). The work
of \cite{DBLP:conf/stoc/DworkTT014} considers the same privacy model
as ours, i.e. $\left(\epsilon,\delta\right)$-DP algorithms under
a single row change with bounded $\ell_{2}$ norm. The works of \cite{DBLP:conf/stoc/HardtR13,DBLP:conf/nips/HardtP14}
study a different privacy model, i.e. $\left(\epsilon,\delta\right)$-DP
algorithms under a change of at most $1$ to a single entry of the
matrix. By applying the result of \cite{DBLP:conf/stoc/HardtR13}
$d$ times and using advanced composition, we obtain an algorithm
for our setting; note that the privacy parameter needs to be set to
$\widetilde{O}\left(\frac{\epsilon}{\sqrt{d}}\right)$ in order to
apply advanced composition, and thus the error increases by a factor
of $\sqrt{d}$ compared to the result of \cite{DBLP:conf/stoc/HardtR13}.
The work of \cite{dOrsiN26} generalizes this model: when considering
$\left(\epsilon,\delta\right)$-DP algorithms under a single row change,
this model measures the change in the $\ell_{1}$ norm. This work
improves the result on \cite{DBLP:conf/stoc/HardtR13} by a factor
of $\frac{\sigma_{1}}{\sigma_{1}-\sigma_{2}}.$ We therefore focus
next on the quantitative comparison with the results by \cite{DBLP:conf/stoc/DworkTT014}
and \cite{dOrsiN26}.

The guarantees provided by our algorithm and that of \cite{dOrsiN26}
are in general incomparable to the guarantee of \cite{DBLP:conf/stoc/DworkTT014}.
Our guarantee and that of \cite{dOrsiN26} improve upon the guarantee
of \cite{DBLP:conf/stoc/DworkTT014} in several practically relevant
regimes where $\Upsilon\le\left\Vert U\right\Vert _{\infty}$ is small
and $\kappa$ is not too small, e.g. when $\left\Vert U\right\Vert _{\infty}=\frac{1}{\sqrt{n}}$
as in the case of Gaussian data and constant gap $\kappa$ (see the
discussion below). On the other hand, there are parameter regimes
where the guarantee of \cite{DBLP:conf/stoc/DworkTT014} is stronger
than that of \cite{dOrsiN26} as well as ours. 

\footnotetext[1]{This result is presented in Theorem 1.3 in \cite{dOrsiN26}.}
\footnotetext[2]{\cite{DBLP:conf/stoc/HardtR13,DBLP:conf/nips/HardtP14,dOrsiN26} use a different symmetrization $\Big(\begin{array}{cc}0 & A\\A^{\top} & 0\end{array}\Big)$. These results are for approximating a vector in the span of $\{(0,v),(u,0)\}$ where $u$ is the top left and $v$ is the top right singular vectors of $A$. When using $A^\top A$, we can compute the bound by \cite{dOrsiN26} as  $\widetilde{O}\Big(\frac{d}{\epsilon^{2}\left(\sigma^{2}_{1}-\sigma^{2}_{2}\right)^{2}}\Big)$, matching the bound of \cite{DBLP:conf/stoc/DworkTT014}.}

The simplest result in the general setting to compare is for $\sin^{2}\left(v_{1},x^{(T)}\right)$,
which is $\widetilde{O}\Big(\frac{d}{\epsilon^{2}\left(\sigma^{2}_{1}-\sigma^{2}_{2}\right)^{2}}\Big)$
for \cite{DBLP:conf/stoc/DworkTT014} (Thm. 6) and $\widetilde{O}\Big(\frac{d\left(\sigma_{1}\Upsilon\right)^{2}+\left(\sigma_{1}\Upsilon\right)^{2}/\kappa}{\epsilon^{2}\left(\sigma^{2}_{1}-\sigma^{2}_{2}\right)^{2}}\Big)$
for ours. Recall that $\sigma_{1}\Upsilon\le1$. If $\kappa>1/d$,
our error is $\widetilde{O}\Big(\frac{d\left(\sigma_{1}\Upsilon\right)^{2}}{\epsilon^{2}\left(\sigma^{2}_{1}-\sigma^{2}_{2}\right)^{2}}\Big)$,
which is smaller than the error of \cite{DBLP:conf/stoc/DworkTT014}
by a $\widetilde{O}((\sigma_{1}\Upsilon)^{2})$ factor.  If $\kappa<1/d$,
the bound of \cite{DBLP:conf/stoc/DworkTT014} is stronger. Intuitively,
this is because the power method needs $1/\kappa$ iterations and
thus, when $1/\kappa>d$, it is better to privatize and release the
entire $d$-dimensional input matrix at once than to probe the input
matrix $1/\kappa>d$ times.

An important setting are instances sampled from random matrix models
such as those considered by the influential work of \cite{DBLP:journals/tit/CandesT10}.
As described in \cite{DBLP:journals/tit/CandesT10}, for many random
matrix models, the matrix has $\left\Vert U\right\Vert _{\infty}=\widetilde{O}\Big(\frac{1}{\sqrt{n}}\Big)$
with high probability. {For the Gaussian row model, the bound we need is instead $\Upsilon=\widetilde O(1/\sqrt n)$, as established in Lemma \ref{lem:gaussian-coherence}.} Let us also consider choosing the
singular values as challenging as possible---the case when $\sigma_{1}=\widetilde{\Theta}\Big(\sqrt{\frac{n}{d}}\Big)$---while
having a constant gap $\sigma_{2}=\frac{1}{2}\sigma_{1}$. In this
setting, the algorithms of \cite{DBLP:conf/stoc/DworkTT014} and \cite{dOrsiN26}
have error $\widetilde{\Omega}\Big(\frac{d^{3}}{\epsilon^{2}n^{2}}\Big)$
and $\widetilde{\Omega}\Big(\frac{d^{2}}{\epsilon^{2}n}\Big)$, whereas
our algorithm has a better error of $\widetilde{O}\Big(\frac{d^{2}}{\epsilon^{2}n^{2}}\Big)$.
We note that an error of $\widetilde{\Omega}\Big(\frac{d^{3}}{\epsilon^{2}n^{2}}\Big)$
or $\widetilde{\Omega}\Big(\frac{d^{2}}{\epsilon^{2}n}\Big)$ is not
meaningful even when $n<\frac{d^{3/2}}{\epsilon}$ or $n<\frac{d^{2}}{\epsilon^{2}}$,
which means our bound allows for a better range of $n$.

\textbf{Gaussian i.i.d. case}. The best known guarantee for private
PCA with i.i.d. samples is achieved by the work of \cite{DBLP:conf/nips/LiuK0O22}.
This work considers a more general setting where the data points are
matrices of arbitrary rank sampled from a sub-Gaussian distribution.
We consider a special but a very natural setting where the rows of
the input matrix are i.i.d. samples from a Gaussian distribution.
When applied to Gaussian data sampled from $\mathcal{N}(0,\overline{\Sigma}^{2})$
with $\sum_{i}\overline{\sigma}^{2}_{i}=1$ as in our Theorem \ref{thm:gaussian},
the work of \cite{DBLP:conf/nips/LiuK0O22} provides two results.
Both results guarantee error $\sin^{2}\left(x,\overline{v}_{1}\right)\leq\widetilde{O}\left(\frac{1}{\overline{\sigma}^{2}_{1}n\overline{\kappa}^{2}}+\frac{d^{2}}{\epsilon^{2}n^{2}\overline{\kappa}^{2}}\right)=\widetilde{O}\left(\frac{d}{n\overline{\kappa}^{2}}+\frac{d^{2}}{\epsilon^{2}n^{2}\overline{\kappa}^{2}}\right)$,
which is optimal for Gaussian i.i.d. data, but only for restricted
parameter regimes. The first result (cf. Lemma 3.1 and Thm. 3.3) requires
$n\geq\widetilde{\Theta}\left(\frac{1}{\overline{\sigma}^{2}_{1}\overline{\kappa}}+\frac{1}{\overline{\sigma}^{2}_{1}\overline{\kappa}^{2}}+\frac{d}{\epsilon\overline{\kappa}}\right)$
and the privacy parameter $\epsilon=O\left(\frac{1}{\sqrt{n}}\right)$.
In contrast, in applications, $\epsilon$ is typically a fixed constant
close to $1$. The second result (cf. Thm. 5.1) applies to all $\epsilon$
but it requires a larger $n\geq\widetilde{\Theta}\left(\exp\left(\frac{1}{\overline{\kappa}^{2}}\right)+\frac{1}{\overline{\sigma}^{2}_{1}\overline{\kappa}^{2}}+\frac{d}{\epsilon\overline{\kappa}}\right)$,
which is prohibitive for even mildly small $\overline{\kappa}$. Our
work achieves error $\widetilde{O}\left(\frac{d}{\epsilon^{2}n^{2}\overline{\sigma}^{2}_{1}\overline{\kappa}^{2}}+\frac{1}{\epsilon^{2}n^{2}\overline{\sigma}^{2}_{1}\overline{\kappa}^{3}}+\frac{1}{n\overline{\sigma}^{2}_{1}\overline{\kappa}^{2}}+\frac{1}{n^{2}\overline{\sigma}^{4}_{1}\overline{\kappa}^{2}}\right)=\widetilde{O}\left(\frac{d}{n\overline{\kappa}^{2}}+\frac{d^{2}}{\epsilon^{2}n^{2}\overline{\kappa}^{2}}+\frac{d}{\epsilon^{2}n^{2}\overline{\kappa}^{3}}\right)$,
since $\overline{\sigma}^{2}_{1}=\Omega(1/d)$. This bound is the
same as the prior when the eigengap $\overline{\kappa}\ge\Omega(\min\left\{ \frac{1}{d},\frac{1}{n\epsilon^{2}}\right\} )$,
{and applies to total privacy budgets $\epsilon\in(0,1)$ with the polynomial sufficient sample size}
 $n\geq{\widetilde{\Theta}\left(d+\frac{d}{\overline{\kappa}^{2}}+\frac{d}{\epsilon\overline{\kappa}^{3/2}}\right)}$.
When the eigengap is small, $\overline{\kappa}\le\min\left\{ \frac{1}{d},\frac{1}{n\epsilon^{2}}\right\} $,
our bound is worse than the error of \cite{DBLP:conf/nips/LiuK0O22}
only for exponentially many samples $n>\exp\left(\frac{1}{\overline{\kappa}^{2}}\right)$).
In this case, our error is only worse by a factor $\overline{\kappa}$,
which is logarithmic in $n$. We note that it is mentioned in \cite{DBLP:conf/nips/LiuK0O22}
that the exponential sample complexity can be improved at the cost
of an extra $\frac{d^{2}}{\epsilon^{2}n^{2}\overline{\kappa}^{4}}$
term in the $\sin$-error, or equivalently an $\frac{d^{4}}{\epsilon^{4}n^{4}\overline{\kappa}^{8}}$
term in the $\sin^{2}$-error, which is worse than our bound in all
regimes.

\textbf{Other related work.} Most relevant to our work are the works
on private algorithms for approximating the top singular vectors \cite{DBLP:conf/stoc/DworkTT014,DBLP:conf/nips/MangoubiV22a,MangoubiV25,TranVV26},
private algorithms for low coherence matrices \cite{DBLP:conf/stoc/HardtR13,DBLP:conf/nips/HardtP14,NicolasSMMC25,dOrsiN26},
private power iteration methods \cite{DBLP:conf/stoc/HardtR13,DBLP:conf/colt/BalcanDWY16,DBLP:conf/nips/HardtP14,NicolasSMMC25},
and algorithms for private estimation of top-$k$ singular vectors
of i.i.d. data \cite{DBLP:conf/nips/LiuK0O22,DunglerS25}. We review
these works in more detail in the appendix.

\subsection{Our algorithm and techniques}

We consider the problem with the input matrix $A\in\R^{n\times d}$.
Our algorithm follows the power iteration template with Gaussian noise
added to each update, ie
\begin{align*}
x^{(t+1)} & =A^{(t)\top}A^{(t)}x^{(t)}+\text{ noise}.
\end{align*}
In prior works, $A^{(t)}$ is the same as the input matrix $A$ while
in our algorithm $A^{(t)}$ differ from $A$. Our algorithm departs
from the algorithm by \cite{DBLP:conf/stoc/HardtR13} due to the change
in the differential privacy model. In the work of \cite{DBLP:conf/stoc/HardtR13},
neighboring input matrices (in this case, $A^{\top}A$ and $A'^{\top}A'$)
differ in one entry by at most $1$. This model allows to bound the
sensitivity of the update by $\left\Vert x^{(t)}\right\Vert _{\infty}$.
From here, imposing a threshold on $\left\Vert x^{(t)}\right\Vert _{\infty}$
is sufficient to guarantee the privacy of the update. On the other
hand, in our model, neighboring inputs ($A$ and $A'$) differ by
a row. The $\ell_{2}$-sensitivity of the power iteration is now $\max_{a}\left\Vert a\right\Vert _{2}\left|\left\langle a,x^{(t)}\right\rangle \right|$,
which grows with $\left\Vert a\right\Vert _{2}$. Hence, simply thresholding
$x^{(t)}$ does not ensure the boundedness of the sensitivity. Instead,
the main idea of our algorithm is to filter out in each iteration
datapoints with large value of $\left\Vert a\right\Vert _{2}\left|\left\langle a,x^{(t)}\right\rangle \right|$.
This step brings two challenges:

\textbf{Determining the filtering threshold}. A naive way is to bound
$\left\Vert a\right\Vert _{2}\left|\left\langle a,x^{(t)}\right\rangle \right|\le\left\Vert a\right\Vert ^{2}_{2}\left\Vert x^{(t)}\right\Vert _{2}$
and set the filtering threshold based on $\left\Vert x^{(t)}\right\Vert _{2}=1$
(if the algorithm normalizes $\left\Vert x^{(t)}\right\Vert _{2}$).
This guarantees the privacy and still allows us to bound the utility
of the algorithm when $\left\Vert a\right\Vert _{2}\le1$ by adding
Gaussian noise proportion to $1$. However, $\left\Vert x^{(t)}\right\Vert _{2}$
can be much larger than $\left\Vert a\right\Vert _{2}\left|\left\langle a,x^{(t)}\right\rangle \right|$.
Indeed, consider the following thought exercise: suppose that $\left\Vert x^{(t)}\right\Vert _{2}=1$
and $a$ comes from a Gaussian $\N(0,\diag(\sigma^{2}_{1},\dots,\sigma^{2}_{d}))$
independent from $x^{(t)}$ such that $\sigma^{2}_{1}+\dots+\sigma^{2}_{d}=1$
and $\sigma_{1}>\sigma_{i}$ for all $i\neq1$. With high probability
$\left\Vert a\right\Vert _{2}\left|\left\langle a,x^{(t)}\right\rangle \right|$
is bounded by $\tilde{O}(\sqrt{\sum\sigma^{2}_{i}x^{(t)2}_{i}})\le\tilde{O}(\sigma_{1})$
which can be significantly smaller than $\left\Vert x^{(t)}\right\Vert _{2}=1$,
for example if $\sigma_{1}=\frac{2}{\sqrt{d}}.$This means the amount
of noise added proportional to $1$ can be too large, which leads
to bad utility guarantee.

\textbf{Filtering breaks the analysis by \cite{DBLP:conf/stoc/HardtR13}}.
One key technique from \cite{DBLP:conf/stoc/HardtR13} is showing
$(\sign(\left\langle v_{i},x^{(t)}\right\rangle ))^{d}_{i=1}$ is
uniformly distributed in $\{-1,1\}^{n}$ for $\{v_{i}\}$ being an
eigenbasis of the input. This allows to bound $\left\Vert x^{(t)}\right\Vert _{\infty}$
by concentration inequalities. However, updating the iterates by the
filtered matrix introduces a noise term coming from the filtered-out
rows. This noise term breaks the independence and symmetry of the
signs and the analysis from \cite{DBLP:conf/stoc/HardtR13} no longer
works in our case.

To overcome these challenges, we introduce adaptive filtering and
a new approach to analyze the algorithm. To filter the rows with large
value $\left\Vert a\right\Vert _{2}\left|\left\langle a,x^{(t)}\right\rangle \right|$,
the algorithm uses the sparse vector technique to find the suitable
threshold in each iteration. The update becomes $x^{(t+1)}=A^{(t)\top}A^{(t)}x^{(t)}+\frac{\theta^{(t)}\sqrt{2\log\frac{2}{\delta}}}{\epsilon}g^{(t)}$
where $A^{(t)}$ is the filtered matrix, $g^{(t)}\sim\N(0,I)$ and
$\theta^{(t)}$ is the threshold (see more detail in Algorithm \ref{alg:Algorithm}).
This method ensures that the threshold is adaptive to the magnitude
of $\left\Vert a\right\Vert _{2}\left|\left\langle a,x^{(t)}\right\rangle \right|$
and the number of filtered datapoints is small, $\tilde{O}\left(1/\epsilon\right)$.
For practical purpose to avoid numerical issues, $x^{(t)}$ can be
normalized to have length $1$. However, for simplicity of the analysis,
we do not normalize the iterates $x^{(t)}$ in our algorithm.

The fact that the algorithm runs in many iterations brings additional
challenges. When the solution $x$ is initialized randomly, the eigen-basis
of the input and $x$ are uncorrelated. Furthermore, as we explained
above, $x$ is also uncorrelated with the rows of $A$. However, after
the first power iteration update, they start to become correlated.
A serious issue arises: the maximum value of $\left\Vert a\right\Vert _{2}\left|\left\langle a,x^{(t)}\right\rangle \right|$
can grow dramatically and our filtering could become ineffective.
We show that this is not the case, and the correlation between $x$
and the rows of $A$ can only grow slowly, and much slower than the
correlation between $x$ and the top eigenvector. Therefore our filtering
is effective in all iterations and the analysis can proceed similarly
to the analysis of the power method.

We now give the intuition for the analysis. We can expect the component
of $x$ along the top eigenvector grows exponentially with the top
eigenvalue of the input, while the remain components grows significantly
slower. This motivates us to study two sequences $M^{(t)}$ and $m^{(t)}$
defined as follows. Let $A=U\Sigma V^{\top}$ be the compact SVD decomposition
of $A$, $B=AV$ be the input matrix after the change of basis to
the column space of $V$, and $y^{(t)}=V^{\top}x^{(t)}$ be the projection
of the iterate $x^{(t)}$ onto the same space. We define $m^{(t)}=\left|y^{(t)}_{1}\right|$
be the component of $y^{(t)}$ along the top eigenvector, which for
simplicity, we assume to be the first coordinate, and $M^{(t)}=\max_{b\in B}\left\langle y^{(t)}_{-1},b_{-1}\right\rangle $
be the maximum dot product between $y^{(t)}$ and the rows of $B$
in the remaining coordinates\footnote{Here we simplify $M^{(t)}$, assuming that all eigenvalues of the
remaining components are equal. In the analysis, we need to take into
account the difference between them.}. Here, $m^{(t)}$ and $M^{(t)}$ signify the evolution of the iterates
along the top and the remaining components, and we want to show that
$m^{(t)}$ grows with $\sigma^{2t}_{1}$ and $M^{(t)}$ grows significantly
slower. We will track the relation between the two quantities by using
an upper bound for $M^{(t)}$ and a lower bound for $m^{(t)}$. A
simplified recurrence between them is as follows, where we omit the
term $\sigma_{1}\Upsilon$
\begin{align*}
M^{(t+1)} & \le(\sigma^{2}_{2}+\tilde{O}(\frac{1}{\epsilon}))M^{(t)}+\tilde{O}(\frac{\sigma_{1}\Upsilon}{\epsilon})m^{(t)}+\text{ noise}.\\
m^{(t+1)} & \ge(\sigma^{2}_{1}-\tilde{O}(\frac{\sigma_{1}\Upsilon}{\epsilon}))m^{(t)}-\tilde{O}(\frac{1}{\epsilon})M^{(t)}+\text{ noise.}
\end{align*}
These bounds are established using the concentration of Gaussian distributions.
We will show that if the gap $\sigma^{2}_{1}-\sigma^{2}_{2}$ is sufficiently
large, $m^{(t+1)}$ grows as fast as $\sigma^{2t}_{1}$ and $M^{(t+1)}$
grows as $\sigma^{2t}_{2}$ instead. After establishing the growth
rate of $M^{(t)}$, we bound the growth of each individual component
$\left|y^{(t)}_{i}\right|$ using the same recurrence style and finally
show the convergence and the error of the approximation.

\section{Preliminaries}

\textbf{Notation}. Let $A\in\R^{n\times d}$ be the input matrix with
$n\geq d$. Let $A=U\Sigma V^{\top}$be the compact SVD decomposition
of $A$ where $U\in\R^{n\times d}$, $\Sigma=\diag(\sigma_{1},\dots,\sigma_{d})\in\R^{d\times d}$
is the diagonal matrix of singular values of $A$ and $V\in\R^{d\times d}$
is a unitary matrix. We thus have $A^{\top}A=V\Sigma^{2}V^{\top}$.
Without loss of generality, assume that the eigenvalues of $A^{\top}A$
are in the decreasing order, i.e $\sigma^{2}_{1}\ge\dots\ge\sigma^{2}_{d}$.
We let $\left\Vert U\right\Vert _{\infty}=\max_{i,j}\left|U_{ij}\right|$,
$\Upsilon=\max_{i}\left|U_{i,1}\right|$. We note that $\Upsilon\le\left\Vert U\right\Vert _{\infty}\leq\sqrt{\frac{\mu\left(A\right)}{n}}$,
where $\mu\left(A\right)$ is the coherence parameter defined in \cite{DBLP:conf/stoc/HardtR13}.
For Gaussian data sampled from $\mathcal{N}(0,\overline{\Sigma}^{2})$,
for any matrix quantity $z$ (eigenvalue, eigenvector, or eigengap),
we use $z$ to denote that quantity for the matrix $A^{\top}A$, and
we use $\overline{z}$ to denote that quantity for the matrix $\overline{\Sigma}^{2}$.
For a vector $x$, we let $\left\Vert x\right\Vert _{2}$ be the Euclidean
norm of $x$. For a matrix $X$, we let $\left\Vert X\right\Vert _{2}$
be the spectral norm of $X$, i.e., the largest singular value of
$X$. For unit length vectors $u$ and $v$, we let $\theta\left(u,v\right)=\arccos\left(\left\langle u,v\right\rangle \right)$
be the angle between $u$ and $v$, and we use $\sin\left(u,v\right)$
as shorthand for $\sin\left(\theta\left(u,v\right)\right)$.

\textbf{Differential Privacy}. We consider the following privacy model.
\begin{defn}
\label{privacy-defn} We say that two matrices $A,A'$ are neighboring
inputs if $A$ and $A'$ differ by one row. A randomized algorithm
${\cal M}$ is said to be $(\epsilon,\delta)$-DP if, for every pair
of neighboring inputs $A,A'$ and every subset of possible outcomes
${\cal O}$, it holds that
\begin{align*}
\Pr\left[{\cal M}(A)\in{\cal O}\right] & \le e^{\epsilon}\Pr\left[{\cal M}(A')\in{\cal O}\right]+\delta.
\end{align*}
When $\delta=0$, we say the algorithm is $\epsilon$-DP.
\end{defn}
Our algorithm makes use of the Sparse Vector Technique, in particular,
the $\AT$ algorithm \cite{DBLP:journals/fttcs/DworkR14}. This algorithm
takes in a set of queries with sensitivity 1 and a threshold $\theta$
and answers with differential privacy if any query value exceeds the
threshold, by adding Laplace noise to both the threshold and the query
values. We provide the utility guarantee of this algorithm in the
analysis of our algorithm in Section \ref{subsec:Analysis}.

\section{Algorithm}

\begin{algorithm}
\caption{Adaptive Power Method for Differentially Private PCA}

\label{alg:Algorithm}

\begin{algorithmic}[1]

\STATE \textbf{Input}: Matrix $A\in\R^{n\times d}$ with rows $a^{\top}_{i}=A_{i}$,
number of iterations $T$, \textbf{per-iteration} privacy parameters
$\epsilon,\delta$, failure probability $\beta$
{($n\ge\max\{2,d\}$, $T\ge1$, $0<\epsilon,\delta,\beta<1$)}

\STATE Let $R=n+1+\frac{12d}{\epsilon}\log\frac{8d\left(T+1\right)}{\beta\delta}$
and ${J=\lceil6T\log_{2}R\rceil}$

\STATE Let $x^{(0)}\gets\N(0,I)$

\STATE \textbf{for} $t=0\dots T-1$

\STATE $\qquad$$\threshold=n-\frac{6\log\frac{4T(J+2)}{\beta}}{\epsilon}+\Lap\left(\frac{2}{\epsilon}\right)$

\STATE $\qquad$for $k=0\dots J$:

\STATE $\qquad\qquad$Query $Q_{k}=\left|\left\{ a\in A:\left\Vert a\right\Vert _{2}\left|\left\langle a,x^{(t)}\right\rangle \right|\le\frac{2^{k}}{n^{4T}}\right\} \right|$.

\STATE $\qquad\qquad$If $Q_{k}+\Lap\left(\frac{4}{\epsilon}\right)\ge\threshold$:
break

\STATE $\qquad${If no query crossed, return $e_1$.}

\STATE $\qquad$Let $\theta^{(t)}=\frac{2^{k}}{n^{4T}}$ and $A^{(t)\top}A^{(t)}=\sum_{a:\left\Vert a\right\Vert _{2}\left|\left\langle a,x^{(t)}\right\rangle \right|\le\theta^{(t)}}aa^{\top}$

\STATE $\qquad$$x^{(t+1)}=A^{(t)\top}A^{(t)}x^{(t)}+\frac{2\theta^{(t)}\sqrt{2\log\frac{2}{\delta}}}{\epsilon}g^{(t)}$
where $g^{(t)}\sim\N(0,I)$ 

\STATE {If $\|x^{(T)}\|_2=0$, return $e_1$.}

\STATE \textbf{return} $\frac{x^{(T)}}{\left\Vert x^{(T)}\right\Vert _{2}}$

\end{algorithmic}
\end{algorithm}

Our algorithm is based on the power method, given in Algorithm \ref{alg:Algorithm}.
The algorithm initializes $x^{(0)}$ as a random Gaussian vector $\N(0,I)$.
In each iteration, to reduce the sensitivity of the update, we filter
out the rows $a$ of the input matrix $A$ with large value for the
$\left\Vert a\right\Vert _{2}\left|\left\langle a,x^{(t)}\right\rangle \right|$.
To do this, we use the sparse vector technique, in particular, the
$\AT$ algorithm \cite{DBLP:journals/fttcs/DworkR14} (Line 4 - Line
7 in Algorithm \ref{alg:Algorithm}). Each query is simply the count
of the number of rows $a$ of $A$ such that $\left\Vert a\right\Vert _{2}\left|\left\langle a,x^{(t)}\right\rangle \right|$
is below a threshold, which is searched over the powers of $2$ in
the range $[\frac{1}{n^{4T}},\frac{2^{J}}{n^{4T}}]$, where $J$ is
defined in the algorithm Following the $\AT$ algorithm, our algorithm
adds Laplace noise (denoted by $\Lap(\cdot)$) to both the threshold
$\threshold$ and the query values $\left|Q_{k}\right|$. After filtering
out the rows with large $\left\Vert a\right\Vert _{2}\left|\left\langle a,x^{(t)}\right\rangle \right|$,
the $\ell_{2}$ sensitivity of the update is the threshold $\theta^{(t)}$.
We then can update the solution using the new matrix with these rows
filtered, following the usual power iteration method, but with Gaussian
noise added (Line 9):
\begin{align*}
x^{(t+1)} & =A^{(t)\top}A^{(t)}x^{(t)}+\text{gaussian noise},
\end{align*}
where $A^{(t)}$ is the matrix with these rows filtered, for which
$A^{(t)\top}A^{(t)}x^{(t)}$ now a small sensitivity.

\subsection{Privacy Guarantee}

The privacy guarantee of Algorithm \ref{alg:Algorithm} is a straightforward
application of the privacy guarantee of the $\AT$ algorithm (see
e.g. \cite{DBLP:journals/fttcs/DworkR14}) and the Gaussian mechanism
(proof in appendix).
\begin{lem}
\label{lem:privacy-guarantee}Suppose that we run Algorithm \ref{alg:Algorithm}
with per-iteration privacy budget $\epsilon,\delta$. Then the algorithm
is $(\epstot,\deltot)$-DP for $\epstot=2\epsilon\sqrt{2\ln\left(\frac{1}{\delta}\right)T}+2\epsilon\frac{e^{2\epsilon}-1}{e^{2\epsilon}+1}T$
and $\deltot=\left(T+1\right)\delta$.
\end{lem}

\subsection{Utility Guarantee}

In this section, we analyze the utility achieved by our algorithm.
Following previous work \cite{DBLP:conf/stoc/DworkTT014}, we assume
that each row $a$ of the input matrix $A$ has length $\left\Vert a\right\Vert _{2}\le1$.
As noted in \cite{DBLP:conf/stoc/DworkTT014}, this assumption is
necessary for achieving any non-trivial utility guarantees. Due to
space limitations, we give below an overview of our proof and defer
the complete proof to Appendix \ref{subsec:Analysis}.

\textbf{Proof overview.} Recall that $A=U\Sigma V^{\top}$ is the
compact SVD decomposition of $A$. For the analysis, we consider the
change of basis $B=AV$. We also let $y^{(t)}=V^{\top}x^{(t)}$ and
$N^{(t)}=V^{\top}\left(A^{\top}A-A^{(t)\top}A^{(t)}\right)V=\sum_{b\in B\colon\left\Vert b\right\Vert _{2}\left|\left\langle b,y^{(t)}\right\rangle \right|>\theta^{(t)}}bb^{\top}$
where the sum is over the rows of $B$. In other words, $N^{(t)}$
consists of the rows filtered out from the input matrix (after the
change of basis). As a reminder, $\Upsilon=\max_{i}\left|U_{i,1}\right|$,
and thus $\sigma_{1}\Upsilon\le1$. We will also use the following
notation for the analysis. For a vector $x$, and a set of indices
$I$, let $x_{I}$ be the vector consisting of the coordinates in
$I$ of $x$. In particular, let $x_{i}$ be the $i$-th coordinate,
and $x_{-i}$ denote the remaining coordinates. For a matrix $A$,
let $A_{I}$ be the matrix formed by the rows of $A$ whose indices
are in $I$ and $A_{:,I}$ be the matrix formed by the columns of
$A$ whose indices are in $I$. We also let $A_{i}$ denote the $i$-th
row and $A_{:,j}$ the $j$-th column. We use $A_{-i}$ to denote
the matrix obtained by removing row $i$ from $A$ and $A_{:,-j}$
the matrix obtained by removing column $j$ from $A$. We write $a\in A$
to indicate $a^{\top}$ is a row of $A$.

We can write the update as
\begin{align}
y^{(t+1)} & =\left(\Sigma^{2}-N^{(t)}\right)y^{(t)}+\frac{\theta^{(t)}\sqrt{2\log\frac{2}{\delta}}}{\epsilon}V^{\top}g^{(t)}.\label{eq:new-update}
\end{align}
Let us show how to bound $\frac{\left\Vert A^{\top}Ax^{(t)}\right\Vert _{2}}{\left\Vert x^{(t)}\right\Vert _{2}}$
which is the same as $\frac{\left\Vert \Sigma^{2}y^{(t)}\right\Vert _{2}}{\left\Vert y^{(t)}\right\Vert _{2}}$
after the change of basis. To bound this, we bound the term corresponding
to the principal component $\left|y^{(t)}_{1}\right|$ and the remaining
components $\left|y^{(t)}_{i}\right|$ for $i\neq1$. Since directly
bounding these terms is challenging, we use an intermediate term that
combines $y^{(t)}_{i}$ for $i\neq1$: $\max_{b\in B}\left|\left\langle y^{(t)}_{-1},b_{-1}\right\rangle \right|$.
The update allows us to lower bound $\left|y^{(t)}_{1}\right|$ and
upper bound $\left|\left\langle y^{(t)}_{-1},b_{-1}\right\rangle \right|$
for each $b$ as follows:
\begin{align}
\left|\left\langle y^{(t+1)}_{-1},b_{-1}\right\rangle \right| & \le\left|\left\langle \Sigma^{2}_{-1}y^{(t)}_{-1},b_{-1}\right\rangle \right|+\frac{\sqrt{2\log\frac{2}{\delta}}}{\epsilon}\left|\left\langle g^{(t)},V_{:,-1}b_{-1}\right\rangle \right|\theta^{(t)}+\left|\left\langle N^{(t)}_{-1}y^{(t)},b_{-1}\right\rangle \right|,\label{eq:1}\\
\text{and }\left|y^{(t+1)}_{1}\right| & \ge\sigma^{2}_{1}\left|y^{(t)}_{1}\right|-\frac{\sqrt{2\log\frac{2}{\delta}}}{\epsilon}\left|\left\langle g^{t},V_{:,1}\right\rangle \right|\theta^{(t)}-\left|N^{t}_{1}y^{(t)}\right|.\label{eq:2}
\end{align}
For the sake of simplicity, let us assume for now that $\sigma^{2}_{2}=\dots=\sigma^{2}_{d}$,
which entails that the term $\left|\left\langle \Sigma^{2}_{-1}y^{(t)}_{-1},b_{-1}\right\rangle \right|$
becomes $\sigma^{2}_{2}\left|\left\langle y^{(t)}_{-1},b_{-1}\right\rangle \right|$.
The terms that involve the Gaussian $g^{(t)}$, i.e. $\left|\left\langle g^{(t)},V_{:,-1}b_{-1}\right\rangle \right|$
and $\left|\left\langle g^{t},V_{:,1}\right\rangle \right|$ can be
bounded via concentration of the Gaussian distribution (Lemma \ref{lem:gaussian}).
The threshold $\theta^{(t)}$ can be approximately bounded by $\max_{b}\left|\left\langle y^{(t)},b\right\rangle \right|$
by the guarantee of the $\AT$ algorithm. We can also bound the terms
$\left|\left\langle N^{(t)}_{-1}y^{(t)},b_{-1}\right\rangle \right|$
and $\left|N^{t}_{1}y^{(t)}\right|$ by $\max_{b}\left|\left\langle y^{(t)},b\right\rangle \right|$.
Since $\max_{b}\left|\left\langle y^{(t)},b\right\rangle \right|$
can be further bounded by $\max_{b}\left|\left\langle y^{(t)}_{-1},b_{-1}\right\rangle \right|+\sigma_{1}\Upsilon\left|y^{(t)}_{1}\right|$
(due to $\left|b_{1}\right|\le\sigma_{1}\Upsilon$), we reduce the
system (\ref{eq:1})-(\ref{eq:2}) to a system of linear inequalities
that involve the relationship between $\max_{b}\left|\left\langle y^{(t+1)}_{-1},b_{-1}\right\rangle \right|$,
$\left|y^{(t+1)}_{1}\right|$ and $\max_{b}\left|\left\langle y^{(t)}_{-1},b_{-1}\right\rangle \right|$
and $\left|y^{(t)}_{1}\right|$, which we can solve.

After obtaining the bound for $\max_{b}\left|\left\langle y^{(t)}_{-1},b_{-1}\right\rangle \right|$
and $\left|y^{(t)}_{1}\right|$, we need to bound how the length $\left\Vert y^{(t)}\right\Vert _{2}$
grows. Here we can return to the update (\ref{eq:new-update}). To
gain an intuition about the technique, we can consider the update
for each coordinate and bound it similarly to what we have done for
the top coordinate. Specifically
\begin{align*}
\left|y^{(t+1)}_{j}\right| & \le\sigma^{2}_{j}\left|y^{(t)}_{j}\right|+\frac{2\sqrt{\log\frac{2}{\delta}}}{\epsilon}\left|V^{\top}_{j}g^{t}\right|\theta^{(t)}+\left|N^{(t)}_{j}y^{(t)}\right|
\end{align*}
which involves terms that we have bounded. Once again, solving a recurrence
relation, we can explicitly obtain a bound for $\left|y^{(t)}_{j}\right|$.
However, this approach will lead to the appearance of the term $\left\Vert U\right\Vert _{\infty}$
which can be larger than $\Upsilon$. We will show a similar, but
more fine-grained analysis by grouping the coordinates into two groups:
$\Head=\{2,\ldots,{\min\{\lfloor4n/\sigma_1^2\rfloor,d\}}\}$ and $\Tail$ being
the remaining coordinates, for which $\left\Vert y^{(t)}_{\Head}\right\Vert _{2}$
and $\left\Vert y^{(t)}_{\Tail}\right\Vert _{2}$ behave differently
but can be directly bounded using the same technique. The final bound
for $\frac{\left\Vert \Sigma^{2}y^{(T)}\right\Vert _{2}}{\left\Vert y^{(T)}\right\Vert _{2}}$
can be lower bounded by $\frac{\sigma^{2}_{1}\left|y^{(T)}_{1}\right|}{\left\Vert y^{(T)}\right\Vert _{2}}$.

Since we do not assume $\sigma^{2}_{2}=\dots=\sigma^{2}_{d}$, we
need a more fine-grained analysis. Instead of just tracking $\max_{b}\left|\left\langle y^{(t)}_{-1},b_{-1}\right\rangle \right|$,
for $t,k\ge0$, we define
\begin{align*}
M^{(t,k)} & =\max_{b\in B}\left|\left\langle \Sigma^{2k}_{-1}y^{(t)}_{-1},b_{-1}\right\rangle \right|,\quad\text{and }\quad m^{(t)}=\left|y^{(t)}_{1}\right|.
\end{align*}
Here $t$ represents the iteration of the algorithm and $k$ is the
power of $\Sigma_{-1}$ we are examining. We also let ${M^{(0)}=\begin{cases}\max_{0\le k\le T}M^{(0,k)}/\sigma_2^{2k},&\sigma_2>0,\\M^{(0,0)},&\sigma_2=0.\end{cases}}$
{Use $\sigma_2^0=1$. If $\sigma_2=0$, $M^{(t,k)}=0$ for $k\ge1$; those recurrence cases are identities, and the $k=0$ argument below remains valid.}
The analysis will bound $M^{(T,0)}$ and $m^{(T)}$, using the similar
recurrence.

\section{I.I.D. Gaussian Data}

In this section, we consider the case when the rows of matrix $A\in\mathbb{R}^{n\times d}$
are i.i.d. samples from a Gaussian distribution $\N(0,\overline{\Sigma}^{2})$,
where $\overline{\Sigma}^{2}\in\R^{d\times d}$ is a PSD matrix. Recall
that, for any matrix quantity $z$ (eigenvalue, eigenvector, or eigengap),
we use $z$ to denote that quantity for the matrix $A^{\top}A$ as
before, and we use $\overline{z}$ to denote that quantity for the
matrix $\overline{\Sigma}^{2}$. We let $\overline{\sigma}^{2}_{1}\geq\overline{\sigma}^{2}_{2}\geq\dots\geq\overline{\sigma}^{2}_{d}$
be the eigenvalues of $\overline{\Sigma}^{2}$. 

\textbf{Proof overview.} The privacy guarantee of Algorithm \ref{alg:Algorithm}
is shown in Lemma \ref{lem:privacy-guarantee}. We now give an overview
of the analysis for the utility guarantee, i.e., show the upper bound
on the error $\sin^{2}\left(\overline{v}_{1},x\right)$ between the
vector $x$ returned by Algorithm \ref{alg:Algorithm} and the top
singular vector $\overline{v}_{1}$ of $\overline{\Sigma}^{2}$stated
in Theorem \ref{thm:gaussian}. The key difference between this case
and the previous case with deterministic data is that the error is
measured with respect to the underlying distribution, instead of the
sampled data. Another challenge is that, to obtain an explicit guarantee,
we need to obtain a bound for the coherence parameter of the data.

 To upper bound $\sin^{2}\left(\overline{v}_{1},x\right)$, we use
that $\sin^{2}\left(\overline{v}_{1},x\right)\leq2\left(\sin^{2}\left(\overline{v}_{1},v_{1}\right)+\sin^{2}\left(x,v_{1}\right)\right)$
 and bound the latter errors. The first error $\sin^{2}\left(\overline{v}_{1},v_{1}\right)$
is due to the difference between the empirical top eigenvector $v_{1}$
of $A^{\top}A$ and the top eigenvector $\overline{v}_{1}$ of the
covariance $\overline{\Sigma}^{2}$ of the underlying distribution.
We can bound this error using concentration inequalities, specifically
the matrix Bernstein inequality, and Wedin's $\sin$-theorem, similar
to techniques in prior work \cite{DBLP:conf/colt/JainJKNS16}. The
second error $\sin^{2}\left(x,v_{1}\right)$ is exactly the error
of the solution with respect to the empirical data, which we have
bounded in Theorem \ref{thm:main}. To justify that we can use the
theorem, we show that the data when sampled from a gaussian distribution
with $\overline{\sigma}^{2}_{1}+\dots+\overline{\sigma}^{2}_{d}=1$,
has length $\tilde{O}(1)$ with high probability, using concentration
bounds for lengths of Gaussian vector from \cite{laurent2000adaptive}
(Lemma \ref{lem:gaussian-row-length-A}). {Finally, to bound $\Upsilon=\|U_{:,1}\|_\infty$, we use rotational invariance of the top left singular line of $A$. After randomizing its sign, the top left singular vector is uniform on $S^{n-1}$; concentration then bounds $\Upsilon$.}

The complete proof of Theorem \ref{thm:gaussian} is provided in Appendix
\ref{subsec:Analysis-gauss}. 

\section*{Acknowledgments}
We thank Linh Tran for showing us the proof of Lemma \ref{lem:gaussian-coherence}. We used GPT 6 Astra to check the proofs.

\bibliographystyle{plain}
\bibliography{ref}

\appendix

\section{Related Work}

We now give an overview of the prior work on private algorithms for
the problem of approximating the top singular vector, its extension
to top-$k$ singular vectors for any $k\geq1$, matrix completion,
and the matrix approximation problems.

\textbf{Privacy model}. The works of \cite{DBLP:conf/pods/BlumDMN05,DBLP:conf/soda/KapralovT13,DBLP:journals/jmlr/ChaudhuriSS13,DBLP:conf/stoc/DworkTT014,DBLP:conf/nips/MangoubiV22a,DBLP:conf/nips/LiuK0O22,MangoubiV25}
provide both $\left(\epsilon,0\right)$ and $\left(\epsilon,\delta\right)$-DP
guarantees for these problems in the same privacy model as ours (where
neighboring input matrices differ in one row and the difference has
bounded $\ell_{2}$ norm). However, these works are either for worst-case
or random i.i.d. data, and do not give fine-grained guarantees based
on problem structures such as low coherence. The works of \cite{DBLP:conf/kdd/McSherryM09,DBLP:conf/stoc/HardtR13,NicolasSMMC25,dOrsiN26}
consider several different privacy models; the most general one is
studied by \cite{dOrsiN26} where neighboring inputs satisfy $\sqrt{\left\Vert EE^{\top}\right\Vert _{1}}\le1$
where $E$ is the difference of the inputs after symmetrization. When
limited to the change in one row in the input matrix, this model means
difference between neighboring input matrices is at most $1$ measured
in the $\ell_{1}$ norm. In the same privacy model where we allow
a $\ell_{1}$ change between neighboring inputs, the work of \cite{DBLP:conf/stoc/HardtR12}
studies the matrix completion problem, and it provides a fine-grained
utility guarantee based on coherence. In contrast, our work consider
neighboring inputs based on the $\ell_{2}$-norm measure of the difference,
which is a more natural and less restrictive privacy model.

\textbf{The worst-case} \textbf{bounds}. \cite{DBLP:conf/stoc/DworkTT014}
provide the tight worst-case bound via the Gaussian mechanism for
the problem of finding the top singular vector. For the extension
to finding the top-$k$ singular vectors, the work of \cite{TranVV26}
provide the best known bound. When $k=1$, this work achieves the
same bound as \cite{DBLP:conf/stoc/DworkTT014}, but improves the
latter for the case when $k\ge2$. Similar to \cite{DBLP:conf/stoc/DworkTT014},
when $k=1$, the bound by \cite{TranVV26} is better than our guarantee
for problems with small eigengap $\kappa<1/d$, but is worse than
our bound when $\kappa\ge1/d$.

\textbf{Beyond-worst-case} \textbf{bounds}. \cite{DBLP:conf/stoc/HardtR13}
first study algorithms for matrices with low coherence and obtain
beyond-worst-case bounds for finding top-$k$ singular vectors for
all $k\ge1$ using a noisy power method framework. The follow up work
by \cite{DBLP:conf/nips/HardtP14} improves this result for the case
when $k\ge2$. The tightest bound known for low coherence data is
provided by \cite{dOrsiN26}. This work improves the work of \cite{DBLP:conf/stoc/HardtR13,DBLP:conf/nips/HardtP14}
in several key aspects. For $k=1$, the utility bound improves a factor
of $\frac{\sigma_{1}}{\sigma_{1}-\sigma_{2}}$. For $k\ge2$, the
returned subspace has rank exactly $k$ with an improved guarantee
even when compared with the subspace of rank $>k$ returned by \cite{DBLP:conf/stoc/HardtR13,DBLP:conf/nips/HardtP14}'s
algorithms. The algorithm by \cite{dOrsiN26} is based on directly
estimating the coherence of the input followed by a Gaussian mechanism,
as opposed to the private power methods. As noted above, however,
these works use a different and more restrictive privacy model when
considering a single row change. Our work also uses a private power
method but for a more relaxed privacy model and offers improved utility
bounds.

\textbf{Private power methods}. Private power iteration methods have
been studied extensively in the literature. Apart from the works of
\cite{DBLP:conf/stoc/HardtR13,DBLP:conf/nips/HardtP14} which also
provide results for the more general noisy power method framework,
notable works in this line of research include \cite{DBLP:conf/colt/BalcanDWY16,NicolasSMMC25,campbell2025decentralized}.
\cite{DBLP:conf/colt/BalcanDWY16} also study the general noisy power
method and provide an improved utility bound for the case $k\ge2$.
As an application, this work provides an algorithm for private PCA
in the distributed setting. Another result from this work is a gap-independent
bound for finding the top-$k$ subspace, but only for the general
noisy power method algorithm. Similarly, the work of \cite{NicolasSMMC25}
studies private power iteration methods in the distributed setting,
with a privacy model that generalizes that of \cite{DBLP:conf/stoc/HardtR13,DBLP:conf/nips/HardtP14}
(yet still restricted to the bounded $\ell_{1}$ norm of the row change).
In a different vein, \cite{campbell2025decentralized} study private
power methods for finding the top left singular vector (i.e., top
eigenvector of $AA^{\top}$) instead of the top right singular vector
(i.e., top eigenvector of $A^{\top}A$) where each row contributes
a data point. In the example of data that are i.i.d. samples from
a Gaussian, the left eigenvectors are uniformly distributed over $S^{n-1}$,
which means estimating the top left singular vector is not meaningful.
Our work contributes to this line of work on private power methods,
but for the model where each row contributes a data point with bounded
$\ell_{2}$ norm. We only consider the problem of finding the top
$1$ right singular vector, and leaving the generalization to $k\ge2$
for future work.

\textbf{Private estimation for i.i.d data}. \cite{DBLP:conf/nips/LiuK0O22}
study the problem of private estimation of the top singular vector
when the data come from a sub-Gaussian distribution. This setting
supports the natural model where each row is treated as a data point
(rank 1 update). Notably, this work considers a more general setting
of an arbitrary rank update. \cite{DunglerS25} generalize this work
for finding the top $k$ singular vectors. Our work considers only
rank $1$ update but requires better sample complexity. Extensions
to arbitrary rank update and to the case when $k\ge2$ are left as
future directions.

\section{Technical Tools}

The Gaussian mechanism \cite{dwork2006our} is a commonly used approach
to achieve differential privacy.
\begin{thm}[Gaussian mechanism, \cite{dwork2006our}]
 Let $f$ be a function of $\ell_{2}$-sensitivity $k$, i.e, on
any neighboring inputs $x$ and $x'$, $\left\Vert f(x)-f(x')\right\Vert _{2}\le k$.
The algorithm ${\cal M}$ that on an input $x$ adds noise generated
with the Gaussian distribution $\N(0,\sigma^{2}I)$ where $\sigma\ge\frac{k}{\epsilon}\sqrt{2\log\frac{2}{\delta}}$
to $f(x)$ is $(\epsilon,\delta)$-DP.
\end{thm}
In the analysis of our algorithm, we use the following well-known
properties of a Gaussian distribution. The proof can be found, for
example, in \cite{DBLP:conf/stoc/HardtR13}.
\begin{lem}[\cite{DBLP:conf/stoc/HardtR13}]
\label{lem:gaussian} Let $g\sim\N(0,I)\in\R^{d}$ then for $a\in\R^{d}$,
the product $\left\langle g,a\right\rangle $ follows a Gaussian distribution
$\N\left(0,\left\Vert a\right\Vert ^{2}_{2}\right)$. Furthermore,
for $g\sim\N(0,1)\in\R$ and $\gamma\geq0$, we have: (1) $\Pr\left[\left|g\right|\le\gamma\right]\le\sqrt{e\gamma}$,
and (2) $\Pr\left[\left|g\right|\ge\gamma\right]\le\exp\left(-\gamma^{2}/2\right)$.
\end{lem}
We will also use the following concentration result for the squared
length of a Gaussian vector, due to Laurent and Massart \cite{laurent2000adaptive}.
\begin{lem}[\cite{laurent2000adaptive}, Lemma 1]
\label{lem:chi-square}Let $g\sim\N(0,\Sigma^{2})$ where $\Sigma$
has eigenvalues $\lambda_{1},\dots,\lambda_{d}>0$.  For all $\gamma\geq0$,
\begin{align*}
\Pr\left[\left\Vert g\right\Vert ^{2}_{2}-\sum_{i}\lambda^{2}_{i}\ge2\sqrt{\sum_{i}\lambda^{4}_{i}}\sqrt{\gamma}+2\lambda^{2}_{1}\gamma\right] & \le\exp\left(-\gamma\right).
\end{align*}
\end{lem}

\section{Analysis of Algorithm \ref{alg:Algorithm} for Deterministic Data
\label{subsec:Analysis}}

\subsection{Privacy analysis}
\begin{thm}[Advanced composition \cite{KairouzOV17}]
 \label{thm:advanced-composition} For all $\epsilon,\delta$, the
adaptive composition of $k$ algorithms, each of which is $(\epsilon,\delta)$-DP,
is $(\epsilon',\delta')$-private where
\[
\epsilon'=\epsilon\sqrt{2\ln\left(\frac{1}{\delta}\right)k}+\epsilon\frac{e^{\epsilon}-1}{e^{\epsilon}+1}k\quad\text{and}\quad\delta'=\delta(k+1)
\]
\end{thm}
\begin{proof}[Proof of Lemma \ref{lem:privacy-guarantee}]
Consider an iteration $t$ of the algorithm. We condition on the
transcript before iteration $t$, so that $x^{(t)}$ becomes fixed.
In the current iteration, the algorithm executes two steps. The first
step is to determine the threshold $\theta^{(t)}$. Since all queries
$Q_{k}$ has sensitivity $1$ and the algorithm adds Laplace noise
to both the threshold $n-\frac{6\log\frac{4T(J+2)}{\beta}}{\epsilon}$
and the query value $\left|Q_{k}\right|$, this step is an instantiation
of the $\AT$ algorithm, which is $\epsilon$-DP. {Its no-crossing branch returns a fixed vector and is postprocessing.} After releasing
the threshold $\theta^{(t)}$, we fix it and show that the second
step that computes the new solution $x^{(t+1)}$ is $(\epsilon,\delta)$-DP
as follows. Consider the function $f_{x,\theta^{(}}\left(A\right)=\sum_{a:\Vert a\Vert_{2}|\langle a,x\rangle|\leq\theta}aa^{\top}x$
for fixed vectors $x,\theta$. We can show that $f_{x,\theta}$ has
$\ell_{2}$-sensitivy at most $2\theta$ as follows. For two input
matrices $A,$ $A'$ with $A'$ having an extra row $a$, we have
\[
\left\Vert f_{x,\theta}\left(A\right)-f_{x,\theta}\left(A'\right)\right\Vert _{2}=\left\Vert aa^{\top}x\cdot1\left\{ \Vert a\Vert_{2}|\langle a,x\rangle|\leq\theta\right\} \right\Vert =\left\Vert a\right\Vert _{2}\left\langle a,x\right\rangle 1\left\{ \Vert a\Vert_{2}|\langle a,x\rangle|\leq\theta\right\} \leq\theta
\]
Therefore, for two matrices $A$ and $A'$ that differ in a single
row, we have $\left\Vert f_{x,\theta}\left(A\right)-f_{x,\theta}\left(A'\right)\right\Vert _{2}\leq2\theta$,
as needed. Thus, adding noise $\frac{2\theta^{(t)}\sqrt{\log\frac{2}{\delta}}}{\epsilon}g^{(t)}$
where $g^{(t)}\sim\N(0,I)$ ensures that the update is $(\epsilon,\delta)$-DP.

Therefore each iteration of the algorithm is $\left(2\epsilon,\delta\right)$-DP.
Thus it follows from advanced composition Theorem \ref{thm:advanced-composition}
that the algorithm is $\left(\epstot,\deltot\right)$-DP with
\begin{align*}
\epstot & =2\epsilon\sqrt{2\ln\left(\frac{1}{\delta}\right)T}+2\epsilon\frac{e^{2\epsilon}-1}{e^{2\epsilon}+1}T\\
\deltot & =\left(T+1\right)\delta
\end{align*}

.
\end{proof}

\subsection{Utility analysis}

For notational simplicity, in this section, we use $\epsilon,\delta$
to denote the \textbf{per-iteration} budget. As shown in Lemma \ref{lem:privacy-guarantee},
the total privacy budget is $\epstot=2\epsilon\sqrt{2\ln\left(\frac{1}{\delta}\right)T}+2\epsilon\frac{e^{2\epsilon}-1}{e^{2\epsilon}+1}T$
and $\deltot=\left(T+1\right)\delta$.

To start the analysis, we write the relationship between $M^{(t,k)}$
and $m^{(t)}$. For each $b\in B$, we have 
\begin{align}
\left|\left\langle \Sigma^{2k}_{-1}y^{(t+1)}_{-1},b_{-1}\right\rangle \right| & =\ \left|\left\langle \Sigma^{2}_{-1}y^{(t)}_{-1}+\frac{2\sqrt{\log\frac{2}{\delta}}}{\epsilon}V^{\top}_{-1}g^{(t)}\theta^{(t)}-N^{(t)}_{-1}y^{(t)},\Sigma^{2k}_{-1}b_{-1}\right\rangle \right|\nonumber \\
 & \le\ \left|\left\langle y^{(t)}_{-1},\Sigma^{2(k+1)}_{-1}b_{-1}\right\rangle \right|+\frac{2\sqrt{\log\frac{2}{\delta}}}{\epsilon}\left|\left\langle g^{(t)},V_{:,-1}\Sigma^{2k}_{-1}b_{-1}\right\rangle \right|\theta^{(t)}\nonumber \\
 & \ +\left|\left\langle N^{(t)}_{-1}y^{(t)},\Sigma^{2k}_{-1}b_{-1}\right\rangle \right|;\label{eq:bound-1}\\
\left|y^{(t+1)}_{1}\right| & =\left|\sigma^{2}_{1}y^{(t)}_{1}+\frac{2\sqrt{\log\frac{2}{\delta}}}{\epsilon}V^{\top}_{1}g^{(t)}\theta^{(t)}-N^{(t)}_{1}y^{(t)}\right|\nonumber \\
 & \ge\sigma^{2}_{1}\left|y^{(t)}_{1}\right|-\frac{2\sqrt{\log\frac{2}{\delta}}}{\epsilon}\left|\left\langle g^{(t)},V_{:,1}\right\rangle \right|\theta^{(t)}-\left|N^{t}_{1}y^{(t)}\right|.\label{eq:bound-2}
\end{align}
We will follow the plan and bound each term in the equations (\ref{eq:bound-1})-(\ref{eq:bound-2}).

The following lemma controls all threshold searches simultaneously.
\begin{lem}
\label{lem:abovethreshold}
Suppose every row of $A$ has norm at most $1$, and let $R,J$ be
as in Algorithm \ref{alg:Algorithm}. With probability at least
$1-\beta/2$, every search crosses and, simultaneously for $0\le t<T$,
\begin{align*}
\left|\left\{a\in A:
\|a\|_2|\langle a,x^{(t)}\rangle|\le\theta^{(t)}
\right\}\right|
&\ge n-\frac{12}{\epsilon}\log\frac{4T(J+2)}{\beta}.
\end{align*}
Moreover, on the same event,
\begin{align}
\theta^{(t)}
&\le n^{-4T}
+2\max_{a\in A}\|a\|_2|\langle a,x^{(t)}\rangle|
\le n^{-4T}+2\max_{b\in B}|\langle b,y^{(t)}\rangle|.
\label{eq:theta}
\end{align}
\end{lem}
\begin{proof}
There are at most $T(J+2)$ Laplace draws in all searches.
With probability at least $1-\beta/4$, every threshold noise has
magnitude at most
$2\log(4T(J+2)/\beta)/\epsilon$ and every query noise has magnitude
at most $4\log(4T(J+2)/\beta)/\epsilon$.
On this event, an accepted count is at least
$n-12\log(4T(J+2)/\beta)/\epsilon$, and any query with count $n$ crosses.
These statements hold also for the adaptive iterates, since the noise
bounds hold for every draw.

Independently of this calculation, Gaussian coordinate tails and a
union bound over $x^{(0)},g^{(0)},\ldots,g^{(T-1)}$ give, with
probability at least $1-\beta/4$,
\[
\max\{\|x^{(0)}\|_2,\max_{t<T}\|g^{(t)}\|_2\}
\le\sqrt{2d\log(8d(T+1)/\beta)}.
\]
On this event the definition of $R$ implies
\[
\|x^{(0)}\|_2+1\le R,\qquad
\frac{2\sqrt{2\log(2/\delta)}}{\epsilon}\|g^{(t)}\|_2
\le\frac{R-n-1}{3}.
\]
Intersect the two events and induct on $t$. If
$\|x^{(t)}\|_2+1\le R^{t+1}$, then, for $t<T$,
\[
\max_a\|a\|_2|\langle a,x^{(t)}\rangle|
\le\|x^{(t)}\|_2<R^{2T}
\le\frac{2^J}{n^{4T}}.
\]
Hence the grid covers all rows, and the search stops no later than
the first grid point covering them. The factor-two spacing gives
\eqref{eq:theta}. Since
$\|A^{(t)\top}A^{(t)}\|_2\le n$, the update then gives
\[
\|x^{(t+1)}\|_2+1
\le n\|x^{(t)}\|_2
+\frac{R-n-1}{3}\bigl(n^{-4T}+2\|x^{(t)}\|_2\bigr)+1
\le R(\|x^{(t)}\|_2+1).
\]
This closes the induction, including initialization, and proves
simultaneous crossing and both bounds. The total failure probability
is at most $\beta/2$.
\end{proof}
\begin{rem}
\label{remark-theta}
The count and threshold upper bounds in Lemma \ref{lem:abovethreshold}
hold on the same event for every iteration.
\end{rem}

The next step is to bound the terms that involve the Gaussian noise
$g^{(t)}$. We use the following lemma:
\begin{lem}
\label{lem:gaussian-2}
With probability at least $1-\beta/2$, simultaneously for $0\le t<T$,
$j\in[d]$, $0\le k\le T$ and every row $a$ of $A$,
\[
|\langle g^{(t)},V_{:,j}\rangle|
\le2\sqrt{\log\frac{16T(T+1)n}{\beta}},
\]
and
\[
|\langle g^{(t)},V_{:,-1}\Sigma_{-1}^{2k}(V^\top a)_{-1}\rangle|
\le2\sigma_2^{2k}\sqrt{\log\frac{16T(T+1)n}{\beta}}.
\]
\end{lem}
\begin{proof}
The first projections have variance $1$ and the second have variance
at most $\sigma_2^{4k}$, since $\|a\|_2\le1$.
There are at most $2T(T+1)n$ projections. Gaussian tails and a union
bound give the stated probability. Zero-variance projections are
identically zero.
\end{proof}

Finally, we bound the terms that involve the filtered rows $N^{(t)}$
in the Lemma \ref{lem:bound-N}.
\begin{lem}
\label{lem:bound-N}Assume that the number of rows filtered out to
get $N^{(t)}$ is at most $R$. For all rows $b\in B$, 
\begin{align*}
\left|\left\langle N^{(t)}_{-1}y^{(t)},\Sigma^{2k}_{-1}b_{-1}\right\rangle \right| & \le R\sigma^{2k}_{2}\max_{b\in B}\left|\left\langle y^{(t)},b\right\rangle \right|.
\end{align*}
Furthermore, 
\begin{align*}
\left|N^{t}_{1}y^{(t)}\right| & \le R\sigma_{1}\Upsilon\max_{b\in B}\left|\left\langle y^{(t)},b\right\rangle \right|.
\end{align*}
\end{lem}
\begin{proof}
We let ${P^{(t)}\coloneqq\left\{ b\in B:\left\Vert b\right\Vert _{2}\left|\left\langle b,y^{(t)}\right\rangle \right|>\theta^{(t)}\right\},\quad |P^{(t)}|\le R}$
{be the set of rows filtered out to get $N^{(t)}$.} We have
\begin{align*}
\left|\left\langle N^{(t)}_{-1}y^{(t)},\Sigma^{2k}_{-1}b_{-1}\right\rangle \right| & =\left|\sum_{z\in P^{(t)}}\left\langle z_{-1}z^{\top}y^{(t)},\Sigma^{2k}_{-1}b_{-1}\right\rangle \right|=\left|\sum_{z\in P^{(t)}}\left\langle y^{(t)},zz^{\top}_{-1}\Sigma^{2k}_{-1}b_{-1}\right\rangle \right|\\
 & \le\sum_{z\in P^{(t)}}\left|\left\langle y^{(t)},z\right\rangle \right|\underbrace{\left|z^{\top}_{-1}\Sigma^{2k}_{-1}b_{-1}\right|}_{\le\sigma^{2k}_{2}}\\
 & \le\left|P^{(t)}\right|\sigma^{2k}_{2}\max_{b\in B}\left|\left\langle y^{(t)},b\right\rangle \right|,
\end{align*}
where the last inequality comes from $\left|z^{\top}_{-1}\Sigma^{2k}_{-1}b_{-1}\right|\le\sigma^{2k}_{2}\left\Vert z\right\Vert _{2}\left\Vert b\right\Vert _{2}\le\sigma^{2k}_{2}$.
For the second claim, we have $\left|b_{1}\right|\le\sigma_{1}\Upsilon$,
hence
\begin{align*}
\left|N^{t}_{1}y^{(t)}\right| & =\left|\sum_{z\in P^{(t)}}z_{1}z^{\top}y^{(t)}\right|\le\left|P^{(t)}\right|\sigma_{1}\Upsilon\max_{b\in B}\left|\left\langle y^{(t)},b\right\rangle \right|.
\end{align*}
\end{proof}

Having bounded the necessary terms, we now solve the system (\ref{eq:bound-1})-(\ref{eq:bound-2}).

\begin{lem}
\label{lem:recurrence}Let $c_{1}={16\sqrt{\log\frac{2}{\delta}\log\frac{16T(T+1)n}{\beta}}}$,
$c_{2}=12\log\frac{4T(J+2)}{\beta}$. Let $K=c_{1}+c_{2}$. We let
${s\ge0}$ be such that 
\begin{align}
\left(\sigma^{2}_{2}+\frac{K}{\epsilon}+\frac{K}{\epsilon}s\right)s & =s\left(\sigma^{2}_{1}-\frac{K}{\epsilon}\sigma_{1}\Upsilon\right)-\frac{K}{\epsilon}\sigma_{1}\Upsilon\label{eq:equation-s}
\end{align}
and 
\begin{align*}
\alpha & =\sigma^{2}_{2}+\frac{K}{\epsilon}+\frac{K}{\epsilon}s.
\end{align*}
{Here $\alpha\ge K/\epsilon>1$.} Then with probability at least $1-\beta$, for all $t,k\ge0$ such
that $t+k\le T$,
\begin{align*}
M^{(t,k)}-\sigma^{2k}_{2}sm^{(t)} & \le\alpha^{t}\sigma^{2k}_{2}\left(M^{(0)}-sm^{(0)}\right)+\frac{\alpha^{t}-1}{\alpha-1}\frac{\sigma^{2k}_{2}K(1+s)}{\epsilon n^{4T}}.
\end{align*}
\end{lem}
\begin{proof}
We condition on the following event, which happens with probability
at least $1-\beta$ {by intersecting the two events of failure at most $\beta/2$ in Lemmas} \ref{lem:abovethreshold} and \ref{lem:gaussian-2}:

1. For all $t\le T,j\in[d]$: $\left|\left\langle g^{(t)},V_{j}\right\rangle \right|\le2{\sqrt{\log\frac{16T(T+1)n}{\beta}}}$;

2. For all $t,k\le T$ and for all rows $b$ of $B$: $\left|\left\langle g^{(t)},V^{\top}_{-1}\Sigma^{2k}_{-1}b_{-1}\right\rangle \right|\le2\sigma^{2k}_{2}{\sqrt{\log\frac{16T(T+1)n}{\beta}}}$.

3. ${P^{(t)}\coloneqq\left\{ b\in B:\left\Vert b\right\Vert _{2}\left|\left\langle b,y^{(t)}\right\rangle \right|>\theta^{(t)}\right\},\quad |P^{(t)}|\le\frac{c_2}{\epsilon}}$.

To start, by Remark \ref{remark-theta}
\begin{align}
\theta^{(t)} & \le\frac{1}{n^{4T}}+2\max_{b}\left|\left\langle y^{(t)},b\right\rangle \right|.\label{eq:theta-1}
\end{align}
By Lemma \ref{lem:bound-N}, for all rows \textbf{$b$ }of $B$ and
for all $j\in[d],$ , we have
\begin{align}
\left|\left\langle N^{(t)}_{-1}y^{(t)},\Sigma^{2k}_{-1}b_{-1}\right\rangle \right| & \le\frac{c_{2}}{\epsilon}\sigma^{2k}_{2}\max_{b\in B}\left|\left\langle y^{(t)},b\right\rangle \right|;\label{eq:noise}\\
\left|N^{(t)}_{1}y^{(t)}\right| & \le\frac{c_{2}}{\epsilon}\sigma_{1}\Upsilon\max_{b\in B}\left|\left\langle y^{(t)},b\right\rangle \right|.\label{eq:noise-bound-row-j}
\end{align}
Now, to bound $\max_{b\in B}\left|\left\langle y^{(t)},b\right\rangle \right|$,
we have
\begin{align*}
\max_{b\in B}\left|\left\langle y^{(t)},b\right\rangle \right| & \le\max_{b\in B}\left(\left|\left\langle y^{(t)}_{-1},b_{-1}\right\rangle \right|+\left|b_{1}\right|\left|y^{(t)}_{1}\right|\right)\\
 & \le\left(M^{(t,0)}+\sigma_{1}\Upsilon m^{(t)}\right).
\end{align*}
From here we can express the system (\ref{eq:bound-1})-(\ref{eq:bound-2})
with the bounds we have obtained. For (\ref{eq:bound-1}), for all
row $b$ of $B$, using the bound for the Gaussian term as well as
the bounds (\ref{eq:theta-1}), (\ref{eq:noise}), we have
\begin{align*}
 & \left|\left\langle \Sigma^{2k}_{-1}y^{(t+1)}_{-1},b_{-1}\right\rangle \right|\\
\le\  & \left|\left\langle \Sigma^{2}_{-1}y^{(t)}_{-1},\Sigma^{2k}_{-1}b_{-1}\right\rangle \right|+\frac{2\sqrt{\log\frac{2}{\delta}}}{\epsilon}\left|\left\langle g^{(t)},V_{:,-1}\Sigma^{2k}_{-1}b_{-1}\right\rangle \right|\theta^{(t)}+\left|\left\langle N^{(t)}_{-1}y^{(t)},\Sigma^{2k}_{-1}b_{-1}\right\rangle \right|\\
\le\  & \left|\left\langle y^{(t)}_{-1},\Sigma^{2(k+1)}_{-1}b_{-1}\right\rangle \right|+\frac{4\sigma^{2k}_{2}{\sqrt{2\log\frac{16T(T+1)n}{\beta}\log\frac{2}{\delta}}}}{\epsilon}\left(2\max_{b\in B}\left|\left\langle y^{(t)},b\right\rangle \right|+\frac{1}{n^{4T}}\right)+\frac{c_{2}}{\epsilon}\sigma^{2k}_{2}\max_{b\in B}\left|\left\langle y^{(t)},b\right\rangle \right|\\
\le\  & \left|\left\langle y^{(t)}_{-1},\Sigma^{2(k+1)}_{-1}b_{-1}\right\rangle \right|+\frac{\sigma^{2k}_{2}}{\epsilon}\left(c_{1}+c_{2}\right)\left(M^{(t,0)}+\sigma_{1}\Upsilon m^{(t)}+\frac{1}{n^{4T}}\right).
\end{align*}
It follows that 
\begin{align}
M^{(t+1,k)} & \le M^{(t,k+1)}+\frac{\sigma^{2k}_{2}}{\epsilon}K\left(M^{(t,0)}+\sigma_{1}\Upsilon m^{(t)}\right)+\frac{\sigma^{2k}_{2}K}{\epsilon n^{4T}}.\label{eq:bound_M}
\end{align}
We also have by the bound (\ref{eq:noise-bound-row-j}),
\begin{align}
m^{(t+1)} & \ge\sigma^{2}_{1}\left|y^{(t)}_{1}\right|-\frac{2\sqrt{\log\frac{2}{\delta}}}{\epsilon}\left|\left\langle g^{(t)},V_{:,1}\right\rangle \right|\theta^{(t)}-\left|N^{(t)}_{1}y^{(t)}\right|\nonumber \\
 & \ge\sigma^{2}_{1}\left|y^{(t)}_{1}\right|-\frac{4\sqrt{\log\frac{2}{\delta}\log\frac{2Tn}{\beta}}}{\epsilon}\theta^{(t)}-\frac{c_{2}}{\epsilon}\sigma_{1}\Upsilon\max_{b\in B}\left|\left\langle y^{(t)},b\right\rangle \right|\nonumber \\
 & \ge\sigma^{2}_{1}\left|y^{(t)}_{1}\right|-\left(\frac{c_{1}}{\epsilon}+\frac{c_{2}}{\epsilon}\sigma_{1}\Upsilon\right)\left(M^{(t,0)}+\sigma_{1}\Upsilon m^{(t)}\right)-\frac{c_{1}}{\epsilon n^{4T}}\nonumber \\
 & \ge\sigma^{2}_{1}m^{(t)}-\frac{1}{\epsilon}K\left(M^{(t,0)}+\sigma_{1}\Upsilon m^{(t)}\right)-\frac{K}{\epsilon n^{4T}}.\label{eq:bound_m}
\end{align}
Having established Eq. (\ref{eq:bound_M}) and Eq. (\ref{eq:bound_m}),
we are going to show the claim in the lemma by induction. For $t=0$,
by the definition of $M^{(0)}$, 
\begin{align*}
M^{(0,k)}-\sigma^{2k}_{2}sm^{(t)} & \le\sigma^{2k}_{2}\left(M^{(0)}-sm^{(0)}\right),\quad\forall k\le T.
\end{align*}
Suppose that the claim is true for all $t'\le t$. We show that it
also holds for $t'=t+1$. From Eq. (\ref{eq:bound_M}) and Eq. (\ref{eq:bound_m})
we have
\begin{align*}
 & M^{(t+1,k)}-\sigma^{2(k)}_{2}sm^{(t+1)}\\
\le\  & M^{(t,k+1)}+\sigma^{2k}_{2}\left(\frac{K}{\epsilon}+\frac{K}{\epsilon}s\right)M^{(t,0)}\\
 & \ -\sigma^{2k}_{2}\left(s\sigma^{2}_{1}-\frac{s}{\epsilon}K\sigma_{1}\Upsilon-\frac{1}{\epsilon}K\sigma_{1}\Upsilon\right)m^{(t)}+\frac{\sigma^{2k}_{2}}{\epsilon n^{4T}}K(1+s)\\
\overset{(a)}{=}\  & M^{(t,k+1)}+\sigma^{2k}_{2}\left(\frac{K}{\epsilon}+\frac{K}{\epsilon}s\right)M^{(t,0)}-\sigma^{2k}_{2}\left(\sigma^{2}_{2}+\frac{K}{\epsilon}+\frac{K}{\epsilon}s\right)sm^{(t)}+\frac{\sigma^{2k}_{2}}{\epsilon n^{4T}}K(1+s)\\
=\  & (M^{(t,k+1)}-\sigma^{2(k+1)}_{2}sm^{(t)})+\sigma^{2k}_{2}\left(\frac{K}{\epsilon}+\frac{K}{\epsilon}s\right)\left(M^{(t,0)}-sm^{(t)}\right)+\frac{\sigma^{2k}_{2}}{\epsilon n^{4T}}K(1+s)\\
\overset{(b)}{\le}\  & \sigma^{2(k+1)}_{2}\alpha^{t}\left(M^{(0)}-sm^{(0)}\right)+\frac{\alpha^{t}-1}{\alpha-1}\frac{\sigma^{2(k+1)}_{2}}{\epsilon n^{4T}}K(1+s)\\
 & \ +\sigma^{2k}_{2}\left(\frac{K}{\epsilon}+\frac{K}{\epsilon}s\right)\alpha^{t}\left(M^{(0)}-sm^{(0)}\right)+\frac{\alpha^{t}-1}{\alpha-1}\frac{\sigma^{2k}_{2}}{\epsilon n^{4T}}\left(\frac{K}{\epsilon}+\frac{K}{\epsilon}s\right)K(1+s)\\
 & \ +\frac{\sigma^{2k}_{2}}{\epsilon n^{4T}}K(1+s)\\
=\  & \alpha^{t+1}\sigma^{2k}_{2}\left(M^{(0)}-sm^{(0)}\right)+\frac{\sigma^{2k}_{2}}{\epsilon n^{4T}}K(1+s)\left(\alpha\cdot\frac{\alpha^{t}-1}{\alpha-1}+1\right)\\
=\  & \alpha^{t+1}\sigma^{2k}_{2}\left(M^{(0)}-sm^{(0)}\right)+\frac{\alpha^{t+1}-1}{\alpha-1}\frac{\sigma^{2k}_{2}}{\epsilon n^{4T}}K(1+s).
\end{align*}
where $(a)$ is due to the definition of $s$ and $(b)$ is due to
the induction hypothesis. By induction, we can conclude the proof
of Lemma \ref{lem:recurrence}.
\end{proof}

If we assume that $\sigma^{2}_{1}-\frac{K}{\epsilon}\sigma_{1}\Upsilon-\sigma^{2}_{2}-\frac{K}{\epsilon}>0$,
{and $\left(\sigma_1^2-\sigma_2^2-\frac K\epsilon(1+\sigma_1\Upsilon)\right)^2>4\frac{K^2}{\epsilon^2}\sigma_1\Upsilon$,} Equation (\ref{eq:equation-s}) is a quadratic equation that admits
two solutions $s_{1}>s_{2}>0$. If we let $\alpha_{i}=\sigma^{2}_{2}+\frac{K}{\epsilon}+\frac{K}{\epsilon}s_{i}$,
as a direct consequence of Lemma \ref{lem:recurrence}, we can bound
the growth of $\max_{b\in B}\left|\left\langle y^{(T)}_{-1},b_{-1}\right\rangle \right|$
and $\left|y^{(T)}_{1}\right|$.
\begin{lem}
\label{exponential-bound}Assume that $\sigma^{2}_{1}-\frac{K}{\epsilon}\sigma_{1}\Upsilon-\sigma^{2}_{2}-\frac{K}{\epsilon}>0$
{and $\left(\sigma_1^2-\sigma_2^2-\frac K\epsilon(1+\sigma_1\Upsilon)\right)^2>4\frac{K^2}{\epsilon^2}\sigma_1\Upsilon$}
and let $s_{1}>s_{2}>0$ be the solutions to the equation (\ref{eq:equation-s})
and $\alpha_{i}=\sigma^{2}_{2}+\frac{K}{\epsilon}+\frac{K}{\epsilon}s_{i}$
for $i\in\{1,2\}$. Then with probability at least $1-\beta$, for
all $t\le T$
\begin{align*}
s_{1}m^{(t)}-M^{(t,0)} & \ge\alpha^{t}_{1}(s_{1}m^{(0)}-M^{(0)}-\frac{T}{\alpha_{1}}\frac{K(1+s_{1})}{\epsilon n^{4T}});\\
M^{(t,0)}-s_{2}m^{(t)} & \le\alpha^{t}_{2}(M^{(0)}-s_{2}m^{(0)}+\frac{T}{\alpha_{2}}\frac{K(1+s_{2})}{\epsilon n^{4T}}).
\end{align*}
\end{lem}
Lemma \ref{exponential-bound} says that if the initial solution $x^{(0)}$
satisfies $s_{1}m^{(0)}-M^{(0)}-\frac{T}{\alpha_{1}}\frac{K(1+s_{1})}{\epsilon n^{4T}}>0$,
$m^{(t)}=\left|y^{(t)}_{1}\right|$ grows exponentially in $\alpha_{1}$
{while the excess $M^{(t,0)}-s_2m^{(t)}$, with
$M^{(t,0)}=\max_{b\in B}|\langle y_{-1}^{(t)},b_{-1}\rangle|$,
has an upper bound growing exponentially in $\alpha_2$.} If we can ensure this
condition and that the gap between $\alpha_{1}$ and $\alpha_{2}$
is sufficiently large, we can show that the algorithm converges and
outputs an approximate solution to the top component. The following
lemma shows that $\alpha_{1}$ acts as a proxy for $\sigma^{2}_{1}$
and $\alpha_{2}$ for $\sigma^{2}_{2}$ in the sense that the gap
between $\alpha_{1}$ and $\alpha_{2}$ is approximately the same
as the gap between $\sigma^{2}_{1}$ and $\sigma^{2}_{2}$.
\begin{lem}
\label{lem:bound-rate}Assuming that $\kappa=\frac{\sigma^{2}_{1}-\sigma^{2}_{2}}{\sigma^{2}_{1}}\ge4\left(\frac{K}{\epsilon\sigma_{1}}\Upsilon+\frac{K^{2}}{\epsilon\sigma^{2}_{1}}\right)$
then $s_{1}=\Omega(\sqrt{\log Tn})$, $s_{2}\le\sigma_{1}\Upsilon$
and $\frac{\alpha_{1}}{\alpha_{2}}\ge1+\frac{\kappa}{2}$.

\end{lem}
\begin{proof}
We first solve for $s_{1}$ and $s_{2}$, which are the solutions
to Equation (\ref{eq:equation-s}). After rearranging the terms, the
equation becomes
\begin{align*}
\frac{K}{\epsilon}s^{2}-\left(\sigma^{2}_{1}-\frac{K}{\epsilon}\sigma_{1}\Upsilon-\sigma^{2}_{2}-\frac{K}{\epsilon}\right)s+\frac{K}{\epsilon}\sigma_{1}\Upsilon & =0.
\end{align*}
The gap hypothesis, $K\ge1$ and $\sigma_1\Upsilon\le1$ imply
\[
\sigma_1^2-\sigma_2^2-\frac K\epsilon(1+\sigma_1\Upsilon)
\ge\frac K\epsilon(4K+3\sigma_1\Upsilon-1)\ge\frac{3K}{\epsilon}.
\]
Its square is strictly larger than
$4(K/\epsilon)^2\sigma_1\Upsilon$, verifying the assumptions.

Solving this equation, we obtain
\begin{align*}
s_{1} & =\frac{\left(\sigma^{2}_{1}-\frac{K}{\epsilon}\sigma_{1}\Upsilon-\sigma^{2}_{2}-\frac{K}{\epsilon}\right)+\sqrt{\left(\sigma^{2}_{1}-\frac{K}{\epsilon}\sigma_{1}\Upsilon-\sigma^{2}_{2}-\frac{K}{\epsilon}\right)^{2}-4\frac{K^{2}}{\epsilon^{2}}\sigma_{1}\Upsilon}}{\frac{2K}{\epsilon}}\\
 & \ge\frac{\left(\sigma^{2}_{1}-\frac{K}{\epsilon}\sigma_{1}\Upsilon-\sigma^{2}_{2}-\frac{K}{\epsilon}\right)}{\frac{K}{\epsilon}}-\frac{\frac{2K}{\epsilon}\sigma_{1}\Upsilon}{\left(\sigma^{2}_{1}-\frac{K}{\epsilon}\sigma_{1}\Upsilon-\sigma^{2}_{2}-\frac{K}{\epsilon}\right)};\\
s_{2} & =\frac{\left(\sigma^{2}_{1}-\frac{K}{\epsilon}\sigma_{1}\Upsilon-\sigma^{2}_{2}-\frac{K}{\epsilon}\right)-\sqrt{\left(\sigma^{2}_{1}-\frac{K}{\epsilon}\sigma_{1}\Upsilon-\sigma^{2}_{2}-\frac{K}{\epsilon}\right)^{2}-4\frac{K^{2}}{\epsilon^{2}}\sigma_{1}\Upsilon}}{\frac{2K}{\epsilon}}\\
 & \le\frac{\frac{2K}{\epsilon}\sigma_{1}\Upsilon}{\left(\sigma^{2}_{1}-\frac{K}{\epsilon}\sigma_{1}\Upsilon-\sigma^{2}_{2}-\frac{K}{\epsilon}\right)}.
\end{align*}
where for the two inequalities we use $\sqrt{a^{2}-b^{2}}\ge a-\frac{b^{2}}{a}$.
From here we can bound
\begin{align}
\alpha_{1} & =\sigma^{2}_{2}+\frac{K}{\epsilon}+\frac{K}{\epsilon}s_{1}\ge\sigma^{2}_{1}-\frac{K}{\epsilon}\sigma_{1}\Upsilon-\frac{\frac{2K^{2}}{\epsilon^{2}}\sigma_{1}\Upsilon}{\left(\sigma^{2}_{1}-\frac{K}{\epsilon}\sigma_{1}\Upsilon-\sigma^{2}_{2}-\frac{K}{\epsilon}\right)};\label{eq:alpha1}\\
\alpha_{2} & =\sigma^{2}_{2}+\frac{K}{\epsilon}+\frac{K}{\epsilon}s_{2}\le\sigma^{2}_{2}+\frac{K}{\epsilon}+\frac{\frac{2K^{2}}{\epsilon^{2}}\sigma_{1}\Upsilon}{\left(\sigma^{2}_{1}-\frac{K}{\epsilon}\sigma_{1}\Upsilon-\sigma^{2}_{2}-\frac{K}{\epsilon}\right)}.\label{eq:alpha2}
\end{align}
Using the condition $\sigma^{2}_{1}-\sigma^{2}_{2}\ge4\left(\frac{K}{\epsilon}\sigma_{1}\Upsilon+\frac{K^{2}}{\epsilon}\right)$,
we have 
\begin{align}
\frac{\frac{2K}{\epsilon}\sigma_{1}\Upsilon}{\left(\sigma^{2}_{1}-\frac{K}{\epsilon}\sigma_{1}\Upsilon-\sigma^{2}_{2}-\frac{K}{\epsilon}\right)} & \le\frac{\frac{2K}{\epsilon}\sigma_{1}\Upsilon}{2\frac{K}{\epsilon}\sigma_{1}\Upsilon}=1.\label{eq:bound-3}
\end{align}
Therefore 
\begin{align*}
s_{1} & {\ge4K-2}=\Omega(\sqrt{\log Tn})
\end{align*}
due to $K=\Omega(\sqrt{\log Tn})$. This establishes the first claim.
For the second claim, we bound the ratio $\frac{\alpha_{1}}{\alpha_{2}}$.
Note that this ratio signifies the convergence rate of the algorithm.
From the bounds (\ref{eq:alpha1}), (\ref{eq:alpha2}), and \ref{eq:bound-3}
\begin{align*}
\frac{\alpha_{1}}{\alpha_{2}} & \ge\frac{\sigma^{2}_{1}-\frac{K}{\epsilon}\sigma_{1}\Upsilon-\frac{2K}{\epsilon}}{\sigma^{2}_{2}+\frac{K}{\epsilon}+\frac{2K}{\epsilon}}=1+\frac{\sigma^{2}_{1}-\sigma^{2}_{2}-\frac{K}{\epsilon}\sigma_{1}\Upsilon-\frac{5K}{\epsilon}}{\sigma^{2}_{2}+\frac{3K}{\epsilon}}\\
 & \ge1+\frac{\sigma^{2}_{1}-\sigma^{2}_{2}}{2\sigma^{2}_{1}}=1+\frac{\kappa}{2}.
\end{align*}
\end{proof}

The following lemma examines the condition on the initial solution
of the algorithm.

\begin{lem}
\label{lem:initial}{Under the hypotheses of Lemma \ref{lem:bound-rate},} with probability at least $\frac{3}{4}$
\begin{align*}
s_{1}m^{(0)}-M^{(0)}-\frac{T}{\alpha_{1}}\frac{K(1+s_{1})}{\epsilon n^{4T}} & >\frac{s_{1}}{1000}.
\end{align*}
\end{lem}
\begin{proof}
First, Gaussian tails give us
\[
\Pr\left[M^{(0)}>\sqrt{2\log(16n(T+1))}\right]\le\frac18.
\]
Also,
$\Pr[|y_1^{(0)}|<1/10]\le\frac1{10}\sqrt{2/\pi}<0.08$.
Thus the two complementary events hold together with probability
greater than $3/4$.

The bound $s_1\ge4K-2$ and the definition of $K$ give
$s_1\ge20\sqrt{2\log(16n(T+1))}$.
Moreover, $\alpha_1=\sigma_2^2+(K/\epsilon)(1+s_1)$ implies
\[
\frac{T}{\alpha_1}\frac{K(1+s_1)}{\epsilon n^{4T}}
\le Tn^{-4T}\le\frac{s_1}{100},
\]
using $n\ge2$ and $T\ge1$.
On the two initialization events,
\[
s_1m^{(0)}-M^{(0)}
-\frac{T}{\alpha_1}\frac{K(1+s_1)}{\epsilon n^{4T}}
\ge\left(\frac1{10}-\frac1{20}-\frac1{100}\right)s_1
>\frac{s_1}{1000}.
\]
\end{proof}

While we have shown that with the right initial solution, $M^{(t,0)}$
grows exponentially slower than $m^{(t)}$, we also need to show that
the length of the solution $\left\Vert y^{(T)}\right\Vert _{2}$ can
be primarily attributed to the top component $\left|y^{(T)}_{1}\right|$.
We show that $\left|y^{(T)}_{1}\right|$ dominates the growth of the
other components. The next lemmata will bound the coordinates together.

\begin{lem}
\label{lem:bound-all-coordinates}Let $\Head=\{2,\ldots,{\min\{\lfloor4n/\sigma_1^2\rfloor,d\}}\}$.
With probability at least $1-3\beta$, we have for all $t\le T$
\begin{align*}
\left\Vert y^{(t)}_{\Head}\right\Vert _{2}\le & r_{\Head}m^{(t)}+\sigma^{2t}_{2}\left\Vert y^{(0)}_{\Head}\right\Vert _{2}\\
 & \qquad+\frac{Kt\alpha^{t-1}_{2}}{\epsilon}\left((\sqrt{1+12\left|\Head\right|}+r_{\Head})C+\frac{r_{\Head}+\sqrt{3\left|\Head\right|}+\frac{1}{2}}{n^{4T}}\right).
\end{align*}
where $r_{\Head}=\left(1+\sqrt{12\left|\Head\right|\kappa}\right)\cdot\frac{\frac{K}{\epsilon}(s_{2}+\sigma_{1}\Upsilon)}{\sigma^{2}_{1}-\sigma^{2}_{2}-\frac{K}{\epsilon}(s_{2}+\sigma_{1}\Upsilon)}$,
and  $C=1+M^{(0)}+\frac{T}{\alpha_{2}}\frac{K(1+s_{2})}{\epsilon n^{4T}}$.
\end{lem}
\begin{proof}
We condition on the following probabilities:

1. From Lemma \ref{lem:chi-square}, with probability at least $1-\beta$,
for all $t\le T$, 
\begin{align}
\left\Vert \frac{2\sqrt{\log\frac{2}{\delta}}}{\epsilon}V^{\top}_{\Head}g^{(t)}\theta^{(t)}\right\Vert ^{2}_{2} & \le\frac{4\log\frac{2}{\delta}}{\epsilon^{2}}\left(\theta^{(t)}\right)^{2}5\left|\Head\right|\log\frac{2T}{\beta}.\label{eq:3-1}
\end{align}

2. From Lemma \ref{lem:gaussian}, $\left\langle \Sigma^{2}_{\Head}y^{(t)}_{\Head}-N^{(t)}_{\Head}y^{(t)},V^{\top}_{\Head}g^{(t)}\right\rangle $$\sim\N\left(0,\left\Vert \Sigma^{2}_{\Head}y^{(t)}_{\Head}-N^{(t)}_{\Head}y^{(t)}\right\Vert ^{2}_{2}\right)$,
so with probability at least $1-\beta$, for all $t\le T$,
\begin{align}
\left|\left\langle \Sigma^{2}_{\Head}y^{(t)}_{\Head}-N^{(t)}_{\Head}y^{(t)},V^{\top}_{\Head}g^{(t)}\right\rangle \right| & \le\left\Vert \Sigma^{2}_{\Head}y^{(t)}_{\Head}-N^{(t)}_{\Head}y^{(t)}\right\Vert _{2}\sqrt{2\log\frac{2T}{\beta}}.\label{eq:3-2}
\end{align}

3. From Lemma \ref{exponential-bound}, with probability $1-\beta$,
we can bound for all $t\le T$
\begin{align*}
\max_{b\in B}\left|\left\langle y^{(t)},b\right\rangle \right|\le & M^{(t,0)}+\sigma_{1}\Upsilon m^{(t)}\\
\le & s_{2}m^{(t)}+\alpha^{t}_{2}(M^{(0)}-s_{2}m^{(0)}+\frac{T}{\alpha_{2}}\frac{K(1+s_{2})}{\epsilon n^{4T}})+\sigma_{1}\Upsilon m^{(t)}\\
= & (s_{2}+\sigma_{1}\Upsilon)m^{(t)}+\left(M^{(0)}-s_{2}m^{(0)}+\frac{T}{\alpha_{2}}\frac{K(1+s_{2})}{\epsilon n^{4T}}\right)\alpha^{t}_{2}\\
\le & (s_{2}+\sigma_{1}\Upsilon)m^{(t)}+C\alpha^{t}_{2}.
\end{align*}
Using the update rule, we have
\begin{align*}
\left\Vert y^{(t+1)}_{\Head}\right\Vert ^{2}_{2} & =\left\Vert \Sigma^{2}_{\Head}y^{(t)}_{\Head}+\frac{2\sqrt{\log\frac{2}{\delta}}}{\epsilon}V^{\top}_{\Head}g^{(t)}\theta^{(t)}-N^{(t)}_{\Head}y^{(t)}\right\Vert ^{2}_{2}\\
 & =\left\Vert \Sigma^{2}_{\Head}y^{(t)}_{\Head}-N^{(t)}_{\Head}y^{(t)}\right\Vert ^{2}_{2}+\left\Vert \frac{2\sqrt{\log\frac{2}{\delta}}}{\epsilon}V^{\top}_{\Head}g^{(t)}\theta^{(t)}\right\Vert ^{2}_{2}\\
 & \ +\frac{4\sqrt{\log\frac{2}{\delta}}}{\epsilon}\theta^{(t)}\left\langle \Sigma^{2}_{\Head}y^{(t)}_{\Head}-N^{(t)}_{\Head}y^{(t)},V^{\top}_{\Head}g^{(t)}\right\rangle 
\end{align*}
For the first term, first notice that, 
\begin{align}
\left\Vert N^{(t)}_{\Head}y^{(t)}\right\Vert _{2} & =\left\Vert \sum_{z\in P^{(t)}}z_{\Head}z^{\top}y^{(t)}\right\Vert _{2}\nonumber \\
 & \le\left|P^{(t)}\right|\max_{b\in B}\left|\left\langle y^{(t)},b\right\rangle \right|\qquad(\text{since }\left\Vert z_{\Head}\right\Vert _{2}\le1)\nonumber \\
 & \le\frac{c_{2}}{\epsilon}\max_{b\in B}\left|\left\langle y^{(t)},b\right\rangle \right|.\label{eq:3}
\end{align}
Using (\ref{eq:3}) we can bound 
\begin{align*}
\left\Vert \Sigma^{2}_{\Head}y^{(t)}_{\Head}-N^{(t)}_{\Head}y^{(t)}\right\Vert _{2} & \le\left\Vert \Sigma^{2}_{\Head}y^{(t)}_{\Head}\right\Vert _{2}+\left\Vert N^{(t)}_{\Head}y^{(t)}\right\Vert _{2}\\
 & \le\sigma^{2}_{2}\left\Vert y^{(t)}_{\Head}\right\Vert _{2}+\frac{c_{2}}{\epsilon}\max_{b\in B}\left|\left\langle y^{(t)},b\right\rangle \right|.
\end{align*}
For the second and third terms, we use (\ref{eq:3-1}) and (\ref{eq:3-2})
and have 
\begin{align*}
\left\Vert y^{(t+1)}_{\Head}\right\Vert ^{2}_{2} & \le\left\Vert \Sigma^{2}_{\Head}y^{(t)}_{\Head}-N^{(t)}_{\Head}y^{(t)}\right\Vert ^{2}_{2}+\frac{40\log\frac{2}{\delta}\log\frac{2T}{\beta}}{\epsilon^{2}}\left(\theta^{(t)}\right)^{2}\left|\Head\right|\\
 & \ +\frac{8\sqrt{\log\frac{2}{\delta}\log\frac{2T}{\beta}}}{\epsilon}\theta^{(t)}\left\Vert \Sigma^{2}_{\Head}y^{(t)}_{\Head}-N^{(t)}_{\Head}y^{(t)}\right\Vert _{2}\\
 & \le\left(\left\Vert \Sigma^{2}_{\Head}y^{(t)}_{\Head}-N^{(t)}_{\Head}y^{(t)}\right\Vert _{2}+\frac{c_{1}}{2\epsilon}\theta^{(t)}\right)^{2}+\frac{c^{2}_{1}}{\epsilon^{2}}\left(\theta^{(t)}\right)^{2}\left|\Head\right|\\
 & \le\left(\sigma^{2}_{2}\left\Vert y^{(t)}_{\Head}\right\Vert _{2}+\frac{K}{\epsilon}\max_{b\in B}\left|\left\langle y^{(t)},b\right\rangle \right|+\frac{c_{1}}{2\epsilon n^{4T}}\right)^{2}+\frac{c^{2}_{1}}{\epsilon^{2}}\left(\theta^{(t)}\right)^{2}\left|\Head\right|,
\end{align*}
where in the last inequality we use $\theta^{(t)}\le\frac{1}{n^{4T}}+2\max_{b\in B}\left|\left\langle y^{(t)},b\right\rangle \right|$.

We prove by induction that for $r=r_{\Head}=\left(1+\sqrt{12\left|\Head\right|\kappa}\right)\cdot\frac{\frac{K}{\epsilon}((s_{2}+\sigma_{1}\Upsilon)}{\sigma^{2}_{1}-\sigma^{2}_{2}-\frac{K}{\epsilon}(s_{2}+\sigma_{1}\Upsilon)}$
we have 
\begin{align*}
\left\Vert y^{(t)}_{\Head}\right\Vert _{2}\le & r_{\Head}m^{(t)}+\sigma^{2t}_{2}\left\Vert y^{(0)}_{\Head}\right\Vert _{2}\\
 & \qquad+\frac{Kt\alpha^{t-1}_{2}}{\epsilon}\left((\sqrt{1+12\left|\Head\right|}+r_{\Head})C+\frac{r_{\Head}+\sqrt{3\left|\Head\right|}+\frac{1}{2}}{n^{4T}}\right).
\end{align*}
First, note that, by the definition of $r$,
\begin{align}
r_{\Head}\sigma^{2}_{2} & =r_{\Head}\left(\sigma^{2}_{1}-\frac{K}{\epsilon}(s_{2}+\sigma_{1}\Upsilon)\right)-\left(1+\sqrt{12\left|\Head\right|\kappa}\right)\cdot\frac{K}{\epsilon}(s_{2}+\sigma_{1}\Upsilon).\label{eq:r}
\end{align}
Recall that $\max_{b\in B}\left|\left\langle y^{(t)},b\right\rangle \right|\le(s_{2}+\sigma_{1}\Upsilon)m^{(t)}+C\alpha^{t}_{2}$,
and by the induction hypothesis, we have
\begin{align*}
 & \sigma^{2}_{2}\left\Vert y^{(t)}_{\Head}\right\Vert _{2}+\frac{K}{\epsilon}\max_{b\in B}\left|\left\langle y^{(t)},b\right\rangle \right|+\frac{c_{1}}{2\epsilon n^{4T}}\\
\le\  & \sigma^{2}_{2}r_{\Head}m^{(t)}+\sigma^{2}_{2}\max\left\{ \left\Vert y^{(t)}_{\Head}\right\Vert _{2}-r_{\Head}m^{(t)},0\right\} +\frac{K}{\epsilon}\max_{b\in B}\left|\left\langle y^{(t)},b\right\rangle \right|+\frac{c_{1}}{2\epsilon n^{4T}}\\
\le\  & r_{\Head}\left(\sigma^{2}_{1}-\frac{K}{\epsilon}(s_{2}+\sigma_{1}\Upsilon)\right)m^{(t)}+\sigma^{2}_{2}\max\left\{ \left\Vert y^{(t)}_{\Head}\right\Vert _{2}-r_{\Head}m^{(t)},0\right\} +\frac{K}{\epsilon}(s_{2}+\sigma_{1}\Upsilon)m^{(t)}+\frac{KC}{\epsilon}\alpha^{t}_{2}\\
 & \quad-\left(1+\sqrt{12\left|\Head\right|\kappa}\right)\cdot\frac{K}{\epsilon}(s_{2}+\sigma_{1}\Upsilon)m^{(t)}+\frac{c_{1}}{2\epsilon n^{4T}}\\
=\  & r_{\Head}\left(\sigma^{2}_{1}-\frac{K}{\epsilon}(s_{2}+\sigma_{1}\Upsilon)\right)m^{(t)}+\sigma^{2}_{2}\max\left\{ \left\Vert y^{(t)}_{\Head}\right\Vert _{2}-r_{\Head}m^{(t)},0\right\} +\frac{KC}{\epsilon}\alpha^{t}_{2}\\
 & \quad-\sqrt{12\left|\Head\right|\kappa}\cdot\frac{K}{\epsilon}(s_{2}+\sigma_{1}\Upsilon)m^{(t)}+\frac{c_{1}}{2\epsilon n^{4T}}\\
=\  & r_{\Head}\left(\sigma^{2}_{1}-\frac{K}{\epsilon}(s_{2}+\sigma_{1}\Upsilon)\right)m^{(t)}+\sigma^{2}_{2}\max\left\{ \left\Vert y^{(t)}_{\Head}\right\Vert _{2}-r_{\Head}m^{(t)},0\right\} \\
 & \quad+\sqrt{1+12\left|\Head\right|}\frac{KC}{\epsilon}\alpha^{t}_{2}+\frac{c_{1}}{2\epsilon n^{4T}}-\left(\sqrt{1+12\left|\Head\right|}-1\right)\frac{KC}{\epsilon}\alpha^{t}_{2}\\
 & \quad-\sqrt{12\left|\Head\right|\kappa}\cdot\frac{K}{\epsilon}(s_{2}+\sigma_{1}\Upsilon)m^{(t)}
\end{align*}
Let $A$ be the positive term on the RHS, $B=\frac{KC}{\epsilon}\alpha^{t}_{2}$
and $D=\frac{K}{\epsilon}(s_{2}+\sigma_{1}\Upsilon)m^{(t)}$. Note
that 
\begin{align*}
 & A-\left(\sqrt{1+12\left|\Head\right|}-1\right)B\ge r_{\Head}\left(\sigma^{2}_{1}-\frac{K}{\epsilon}(s_{2}+\sigma_{1}\Upsilon)\right)m^{(t)}\\
=\  & \left(1+\sqrt{12\left|\Head\right|\kappa}\right)\cdot\frac{\frac{K}{\epsilon}(s_{2}+\sigma_{1}\Upsilon)}{\sigma^{2}_{1}-\sigma^{2}_{2}-\frac{K}{\epsilon}(s_{2}+\sigma_{1}\Upsilon)}\left(\sigma^{2}_{1}-\frac{K}{\epsilon}(s_{2}+\sigma_{1}\Upsilon)\right)m^{(t)}\\
\ge\  & \left(1+\sqrt{12\left|\Head\right|\kappa}\right)D\cdot\frac{\sigma^{2}_{1}-\frac{K}{\epsilon}(s_{2}+\sigma_{1}\Upsilon)}{\sigma^{2}_{1}-\sigma^{2}_{2}-\frac{K}{\epsilon}(s_{2}+\sigma_{1}\Upsilon)}\\
\ge\  & \frac{1+\sqrt{12\left|\Head\right|\kappa}}{\kappa}D,
\end{align*}
and 
\begin{align*}
A & \ge\sqrt{1+12\left|\Head\right|}B,
\end{align*}
The scalar upper bound being squared is nonnegative: its
$m^{(t)}$ terms combine to
$\bigl(r_{\Head}\sigma_2^2+(K/\epsilon)(s_2+\sigma_1\Upsilon)\bigr)m^{(t)}$,
and all remaining terms are nonnegative.
Squaring both sides, we have 
\begin{align*}
 & \left(\sigma^{2}_{2}\left\Vert y^{(t)}_{\Head}\right\Vert _{2}+\frac{K}{\epsilon}\max_{b\in B}\left|\left\langle y^{(t)},b\right\rangle \right|+\frac{c_{1}}{2\epsilon n^{4T}}\right)^{2}\\
\le\  & \left(A-\left(\sqrt{1+12\left|\Head\right|}-1\right)B-\sqrt{12\left|\Head\right|\kappa}D\right)^{2}\\
{\le}\  & \left(A-\left(\sqrt{1+12\left|\Head\right|}-1\right)B\right)^{2}-\left(2\frac{1+\sqrt{12\left|\Head\right|\kappa}}{\kappa}\sqrt{12\left|\Head\right|\kappa}-12\left|\Head\right|\kappa\right)D^{2}\\
\le\  & \left(A-\left(\sqrt{1+12\left|\Head\right|}-1\right)B\right)^{2}-12\left|\Head\right|D^{2}\\
\le\  & A^{2}-\left(2\sqrt{1+12\left|\Head\right|}\left(\sqrt{1+12\left|\Head\right|}-1\right)-\left(\sqrt{1+12\left|\Head\right|}-1\right)^{2}\right)B^{2}-12\left|\Head\right|D^{2}\\
=\  & A^{2}-12\left|\Head\right|B^{2}-12\left|\Head\right|D^{2}.
\end{align*}
On the other hand we have 
\begin{align*}
\frac{c^{2}_{1}}{\epsilon^{2}}\left(\theta^{(t)}\right)^{2} & \le\frac{c^{2}_{1}}{\epsilon^{2}}\left(\frac{1}{n^{4T}}+2(s_{2}+\sigma_{1}\Upsilon)m^{(t)}+2C\alpha^{t}_{2}\right)^{2}\\
 & \le\frac{3c^{2}_{1}}{\epsilon^{2}n^{8T}}+12B^{2}+12D^{2}.
\end{align*}
We have 
\begin{align*}
\left\Vert y^{(t+1)}_{\Head}\right\Vert ^{2}_{2} & \le\left(\sigma^{2}_{2}\left\Vert y^{(t)}_{\Head}\right\Vert _{2}+\frac{K}{\epsilon}\max_{b\in B}\left|\left\langle y^{(t)},b\right\rangle \right|+\frac{c_{1}}{2\epsilon n^{4T}}\right)^{2}+\frac{c^{2}_{1}}{\epsilon^{2}}\left(\theta^{(t)}\right)^{2}\left|\Head\right|\\
 & \le A^{2}+\frac{3c^{2}_{1}}{\epsilon^{2}n^{8T}}\left|\Head\right|.
\end{align*}
Taking the square root, we have 
\begin{align*}
\left\Vert y^{(t+1)}_{\Head}\right\Vert _{2} & \le r_{\Head}\left(\sigma^{2}_{1}-\frac{K}{\epsilon}(s_{2}+\sigma_{1}\Upsilon)\right)m^{(t)}+\sigma^{2}_{2}\max\left\{ \left\Vert y^{(t)}_{\Head}\right\Vert _{2}-r_{\Head}m^{(t)},0\right\} \\
 & +\sqrt{1+12\left|\Head\right|}\frac{KC}{\epsilon}\alpha^{t}_{2}+\frac{c_{1}}{\epsilon n^{4T}}\left(\sqrt{3\left|\Head\right|}+\frac{1}{2}\right).
\end{align*}
As before 
\begin{align}
m^{(t+1)} & \ge\sigma^{2}_{1}m^{(t)}-\frac{K}{\epsilon}\left((s_{2}+\sigma_{1}\Upsilon)m^{(t)}+C\alpha^{t}_{2}\right)-\frac{K}{\epsilon n^{4T}}\nonumber \\
 & =\left(\sigma^{2}_{1}-\frac{K}{\epsilon}(s_{2}+\sigma_{1}\Upsilon)\right)m^{(t)}-\frac{K}{\epsilon}C\alpha^{t}_{2}-\frac{K}{\epsilon n^{4T}}.\label{eq:m}
\end{align}
Consequently, 
\begin{align*}
\max\left\{ \left\Vert y^{(t+1)}_{\Head}\right\Vert _{2}-r_{\Head}m^{(t+1)},0\right\}  & \le\sigma^{2}_{2}\max\left\{ \left\Vert y^{(t)}_{\Head}\right\Vert _{2}-r_{\Head}m^{(t)},0\right\} \\
 & +\left(\sqrt{1+12\left|\Head\right|}+r_{\Head}\right)\frac{KC}{\epsilon}\alpha^{t}_{2}\\
 & +\frac{K}{\epsilon n^{4T}}\left(\sqrt{3\left|\Head\right|}+\frac{1}{2}+r_{\Head}\right).
\end{align*}
Here we use $\alpha_2\ge\max\{1,\sigma_2^2\}$. Induction gives us the stated bound.

\end{proof}

\begin{lem}
\label{lem:bound-tail}Let $j={\min\{\lfloor4n/\sigma_1^2\rfloor,d\}}+1$ and $\Tail=\{{\min\{\lfloor4n/\sigma_1^2\rfloor,d\}}+1,\ldots,d\}$.
If $j\le d$ then with probability at least $1-3\beta$, we have
\begin{align*}
\left\Vert y^{(t)}_{\Tail}\right\Vert _{2}\le & r_{\Tail}m^{(t)}+\sigma^{2t}_{j}\left\Vert y^{(0)}_{\Tail}\right\Vert _{2}+\frac{Kt\alpha^{t-1}_{2}}{\epsilon}\left(\left(\sqrt{1+12d}+r_{\Tail}\right)C+\frac{r_{\Tail}+\sqrt{3d}+\frac{1}{2}}{n^{4T}}\right).
\end{align*}
where $r_{\Tail}=\left(1+\sqrt{12d}\right)\cdot\frac{\frac{K}{\epsilon}(s_{2}+\sigma_{1}\Upsilon)}{\sigma^{2}_{1}-\sigma^{2}_{j}-\frac{K}{\epsilon}(s_{2}+\sigma_{1}\Upsilon)}$,
and $C=1+M^{(0)}+\frac{T}{\alpha_{2}}\frac{K(1+s_{2})}{\epsilon n^{4T}}.$
\end{lem}
\begin{rem}
Note that in the trivial case when $j>d$, we have $\left\Vert y^{(t)}_{\Tail}\right\Vert _{2}=0$.
\end{rem}
\begin{proof}
Similar to the lemma, we also condition on the following probabilities

1. From Lemma \ref{lem:chi-square}, with probability at least $1-\beta$,
for all $t\le T$, 
\begin{align}
\left\Vert \frac{2\sqrt{\log\frac{2}{\delta}}}{\epsilon}V^{\top}_{\Tail}g^{(t)}\theta^{(t)}\right\Vert ^{2}_{2} & \le\frac{4\log\frac{2}{\delta}}{\epsilon^{2}}\left(\theta^{(t)}\right)^{2}5\left|\Tail\right|\log\frac{2T}{\beta}.\label{eq:4-1}
\end{align}

2. From Lemma \ref{lem:gaussian}, $\left\langle \Sigma^{2}_{\Tail}y^{(t)}_{\Tail}-N^{(t)}_{\Tail}y^{(t)},V^{\top}_{\Tail}g^{(t)}\right\rangle $$\sim\N\left(0,\left\Vert \Sigma^{2}_{\Tail}y^{(t)}_{\Tail}-N^{(t)}_{\Tail}y^{(t)}\right\Vert ^{2}_{2}\right)$,
so with probability at least $1-\beta$, for all $t\le T$,
\begin{align}
\left|\left\langle \Sigma^{2}_{\Tail}y^{(t)}_{\Tail}-N^{(t)}_{\Tail}y^{(t)},V^{\top}_{\Tail}g^{(t)}\right\rangle \right| & \le\left\Vert \Sigma^{2}_{\Tail}y^{(t)}_{\Tail}-N^{(t)}_{\Tail}y^{(t)}\right\Vert _{2}\sqrt{2\log\frac{2T}{\beta}}.\label{eq:4-2}
\end{align}

3. With probability $1-\beta$, we can bound for all $t\le T$
\begin{align*}
\max_{b\in B}\left|\left\langle y^{(t)},b\right\rangle \right|\le & M^{(t,0)}+\sigma_{1}\Upsilon m^{(t)}\le(s_{2}+\sigma_{1}\Upsilon)m^{(t)}+C\alpha^{t}_{2}.
\end{align*}
Let $j={\min\{\lfloor4n/\sigma_1^2\rfloor,d\}}+1$ and assume that $j\le d$. As
before, we can bound
\begin{align*}
\left\Vert y^{(t+1)}_{\Tail}\right\Vert ^{2}_{2} & =\left\Vert \Sigma^{2}_{\Tail}y^{(t)}_{\Tail}+\frac{2\sqrt{\log\frac{2}{\delta}}}{\epsilon}V^{\top}_{\Tail}g^{(t)}\theta^{(t)}-N^{(t)}_{\Tail}y^{(t)}\right\Vert ^{2}_{2}\\
 & =\left\Vert \Sigma^{2}_{\Tail}y^{(t)}_{\Tail}-N^{(t)}_{\Tail}y^{(t)}\right\Vert ^{2}+\left\Vert \frac{2\sqrt{\log\frac{2}{\delta}}}{\epsilon}V^{\top}_{\Tail}g^{(t)}\theta^{(t)}\right\Vert ^{2}_{2}\\
 & \ +\frac{{4}\sqrt{\log\frac{2}{\delta}}}{\epsilon}\theta^{(t)}\left\langle {\Sigma^{2}_{\Tail}}y^{(t)}_{\Tail}-N^{(t)}_{\Tail}y^{(t)},V^{\top}_{\Tail}g^{(t)}\right\rangle .
\end{align*}
For the first term, using, $\left\Vert N^{(t)}_{\Tail}y^{(t)}\right\Vert _{2}\le\frac{c_{2}}{\epsilon}\max_{b\in B}\left|\left\langle y^{(t)},b\right\rangle \right|$,
we have 
\begin{align*}
\left\Vert \Sigma^{2}_{\Tail}y^{(t)}_{\Tail}-N^{(t)}_{\Tail}y^{(t)}\right\Vert _{2} & \le\left\Vert \Sigma^{2}_{\Tail}y^{(t)}_{\Tail}\right\Vert _{2}+\left\Vert N^{(t)}_{\Tail}y^{(t)}\right\Vert _{2}\\
 & \le\sigma^{2}_{j}\left\Vert y^{(t)}_{\Tail}\right\Vert _{2}+\frac{c_{2}}{\epsilon}\max_{b\in B}\left|\left\langle y^{(t)},b\right\rangle \right|.
\end{align*}
The second and third terms come from above, so we have
\begin{align*}
\left\Vert y^{(t+1)}_{\Tail}\right\Vert ^{2}_{2} & \le\left\Vert \Sigma^{2}_{\Tail}y^{(t)}_{\Tail}-N^{(t)}_{\Tail}y^{(t)}\right\Vert ^{2}_{2}+\frac{40\log\frac{2}{\delta}\log\frac{2T}{\beta}}{\epsilon^{2}}\left(\theta^{(t)}\right)^{2}\left|\Tail\right|\\
 & \ +\frac{8\sqrt{\log\frac{2}{\delta}\log\frac{2T}{\beta}}}{\epsilon}\theta^{(t)}\left\Vert \Sigma^{2}_{\Tail}y^{(t)}_{\Tail}-N^{(t)}_{\Tail}y^{(t)}\right\Vert \\
 & \le\left(\left\Vert \Sigma^{2}_{\Tail}y^{(t)}_{\Tail}-N^{(t)}_{\Tail}y^{(t)}\right\Vert _{2}+\frac{c_{1}}{2\epsilon}\theta^{(t)}\right)^{2}+\frac{c^{2}_{1}}{\epsilon^{2}}\left(\theta^{(t)}\right)^{2}d\\
 & \le\left(\sigma^{2}_{j}\left\Vert y^{(t)}_{\Tail}\right\Vert _{2}+\frac{K}{\epsilon}\max_{b\in B}\left|\left\langle y^{(t)},b\right\rangle \right|+\frac{c_{1}}{2\epsilon n^{4T}}\right)^{2}+\frac{c^{2}_{1}}{\epsilon^{2}}\left(\theta^{(t)}\right)^{2}d.
\end{align*}
From here we can follow the same induction argument from Lemma \ref{lem:bound-all-coordinates}
with $\sigma_{2}$ replaced by $\sigma_{j}$ and $\left|\Head\right|$
replaced by $d$.

\end{proof}

Having the bounds for $s_{2}$ from Lemma \ref{lem:bound-rate} and
$\left\Vert y^{(t)}_{\Head}\right\Vert _{2}$, $\left\Vert y^{(t)}_{\Tail}\right\Vert _{2}$,
we can bound $\left\Vert y^{(T)}_{-1}\right\Vert $ explicitly as
follows.
\begin{lem}
\label{lem:bound-y}For $\kappa=\frac{\sigma^{2}_{1}-\sigma^{2}_{2}}{\sigma^{2}_{1}}\ge4\left(\frac{K}{\epsilon\sigma_{1}}\Upsilon+\frac{K^{2}}{\epsilon\sigma^{2}_{1}}\right)$
With probability $1-6\beta$, we have
\begin{align*}
\left\Vert y^{(T)}_{-1}\right\Vert _{2} & \le Rm^{(T)}+O\left(KT\alpha^{T}_{2}\sqrt{d}\right),
\end{align*}
where $R=O\left(\left(\frac{\sqrt{\min\left\{ 4n/\sigma^{2}_{1},d\right\} }}{\epsilon\sigma^{2}_{1}\sqrt{\kappa}}+\frac{1}{\epsilon\sigma^{2}_{1}\kappa}+\frac{\sqrt{d}}{\epsilon\sigma^{2}_{1}}\right)K\sigma_{1}\Upsilon\right)$.
\end{lem}
\begin{proof}
{The common recurrence event and the two norm and
cross-term events for each of the head and tail have total failure
at most $5\beta$. Intersecting them with the joint initialization
bounds below, of failure at most $\beta$, gives probability at least
$1-6\beta$ for all the following estimates.} We have
\begin{align*}
\left\Vert y^{(t)}_{\Head}\right\Vert _{2}\le & r_{\Head}m^{(t)}+\sigma^{2t}_{2}\left\Vert y^{(0)}_{\Head}\right\Vert _{2}\\
 & +\frac{Kt\alpha^{t-1}_{2}}{\epsilon}\left((\sqrt{1+12\left|\Head\right|}+r_{\Head})C+\frac{r_{\Head}+\sqrt{3\left|\Head\right|}+\frac{1}{2}}{n^{4T}}\right).\\
\left\Vert y^{(t)}_{\Tail}\right\Vert _{2}\le & r_{\Tail}m^{(t)}+\sigma^{2t}_{j}\left\Vert y^{(0)}_{\Tail}\right\Vert _{2}\\
 & +\frac{Kt\alpha^{t-1}_{2}}{\epsilon}\left(\left(\sqrt{1+12d}+r_{\Tail}\right)C+\frac{r_{\Tail}+\sqrt{3d}+\frac{1}{2}}{n^{4T}}\right).
\end{align*}
where $j={\min\{\lfloor4n/\sigma_1^2\rfloor,d\}}+1$ and
\begin{align*}
r_{\Head} & =\left(1+\sqrt{12\left|\Head\right|\kappa}\right)\cdot\frac{\frac{K}{\epsilon}(s_{2}+\sigma_{1}\Upsilon)}{\sigma^{2}_{1}-\sigma^{2}_{2}-\frac{K}{\epsilon}(s_{2}+\sigma_{1}\Upsilon)}\\
r_{\Tail} & =\left(1+\sqrt{12d}\right)\cdot\frac{\frac{K}{\epsilon}(s_{2}+\sigma_{1}\Upsilon)}{\sigma^{2}_{1}-\sigma^{2}_{j}-\frac{K}{\epsilon}(s_{2}+\sigma_{1}\Upsilon)}\\
C & =1+M^{(0)}+\frac{T}{\alpha_{2}}\frac{K(1+s_{2})}{\epsilon n^{4T}}.
\end{align*}
Gaussian tails and a union bound give, with probability at
least $1-\beta$,
\[
\|y^{(0)}\|_2\le\sqrt{2d\log(4d/\beta)}\le K\sqrt d,
\qquad
M^{(0)}\le\sqrt{2\log(4n(T+1)/\beta)}\le K.
\]
Recall that for $\kappa=\frac{\sigma^{2}_{1}-\sigma^{2}_{2}}{\sigma^{2}_{1}}\ge4\left(\frac{K}{\epsilon\sigma_{1}}\Upsilon+\frac{K^{2}}{\epsilon\sigma^{2}_{1}}\right)$,
\begin{align*}
s_{2} & \le\frac{\frac{2K}{\epsilon}\sigma_{1}\Upsilon}{\left(\sigma^{2}_{1}-\frac{K}{\epsilon}\sigma_{1}\Upsilon-\sigma^{2}_{2}-\frac{K}{\epsilon}\right)}\le\sigma_{1}\Upsilon\le1.
\end{align*}
If follows that 
\begin{align*}
\sigma^{2}_{1}-\sigma^{2}_{2}-\frac{K}{\epsilon}(s_{2}+\sigma_{1}\Upsilon) & \ge\sigma^{2}_{1}-\sigma^{2}_{2}-\frac{K\sigma_{1}\Upsilon}{\epsilon}-\frac{K}{\epsilon}\\
 & \ge\frac{3K\sigma_{1}\Upsilon}{\epsilon}+\frac{3K}{\epsilon}\\
 & \ge\frac{3K}{\epsilon}(s_{2}+\sigma_{1}\Upsilon).
\end{align*}
So 
\begin{align*}
1+r_{\Head} & \le1+\frac{1+\sqrt{12\left|\Head\right|\kappa}}{3}\le3\sqrt{d}\\
1+r_{\Tail} & \le1+\frac{1+\sqrt{12d}}{3}\le3\sqrt{d}\\
C & \le1+M^{(0)}+\frac{T}{\alpha_{2}}\frac{2K}{\epsilon n^{4T}}\le3K\leq3\epsilon\alpha_{2}.
\end{align*}
Further, since $\alpha_{2}\ge\frac{K}{\epsilon}$ and $\alpha_{2}\ge\sigma^{2}_{2}$,
we obtain 
\begin{align*}
\left\Vert y^{(T)}_{\Head}\right\Vert _{2}\le & r_{\Head}m^{(T)}+\sigma^{2T}_{2}\left\Vert y^{(0)}_{\Head}\right\Vert _{2}\\
 & +\frac{KT\alpha^{T-1}_{2}}{\epsilon}\left((\sqrt{1+12\left|\Head\right|}+r_{\Head})C+O\left(\frac{\sqrt{d}}{n^{4T}}\right)\right)\\
\le & r_{\Head}m^{(T)}+O\left(KT\alpha^{T}_{2}\sqrt{d}\right)\\
\left\Vert y^{(T)}_{\Tail}\right\Vert _{2}\le & r_{\Tail}m^{(T)}+\sigma^{2T}_{j}\left\Vert y^{(0)}_{\Tail}\right\Vert _{2}\\
 & +\frac{KT\alpha^{T-1}_{2}}{\epsilon}\left(\left(\sqrt{1+12d}+r_{\Tail}\right)C+O\left(\frac{\sqrt{d}}{n^{4T}}\right)\right)\\
\le & r_{\Tail}m^{(T)}+O\left(KT\alpha^{T}_{2}\sqrt{d}\right)
\end{align*}
To further bound $\left\Vert y^{(T)}_{\Head}\right\Vert _{2}+\left\Vert y^{(T)}_{\Tail}\right\Vert _{2}$,
we have
\begin{align*}
\sigma^{2}_{1}-\sigma^{2}_{2}-\frac{K}{\epsilon}(s_{2}+\sigma_{1}\Upsilon) & \ge\max\left\{ \frac{\sigma^{2}_{1}-\sigma^{2}_{2}}{2},\frac{2K}{\epsilon}\right\} \\
\sigma^{2}_{1}-\sigma^{2}_{j}-\frac{K}{\epsilon}(s_{2}+\sigma_{1}\Upsilon) & \ge\frac{\sigma^{2}_{1}}{2}
\end{align*}
Hence we can bound
\begin{align*}
r_{\Head} & =\left(1+\sqrt{12\left|\Head\right|\kappa}\right)\cdot\frac{\frac{K}{\epsilon}(s_{2}+\sigma_{1}\Upsilon)}{\sigma^{2}_{1}-\sigma^{2}_{2}-\frac{K}{\epsilon}(s_{2}+\sigma_{1}\Upsilon)}\\
 & \le\left(1+\sqrt{12\left|\Head\right|\kappa}\right)\cdot\frac{\frac{2K\sigma_{1}\Upsilon}{\epsilon}}{\frac{\sigma^{2}_{1}-\sigma^{2}_{2}}{2}}\\
 & =\left(1+\sqrt{12\left|\Head\right|\kappa}\right)\frac{4K\sigma_{1}\Upsilon}{\epsilon\sigma^{2}_{1}\kappa}.\\
r_{\Tail} & =\left(1+\sqrt{12d}\right)\cdot\frac{\frac{K}{\epsilon}(s_{2}+\sigma_{1}\Upsilon)}{\sigma^{2}_{1}-\sigma^{2}_{j}-\frac{K}{\epsilon}(s_{2}+\sigma_{1}\Upsilon)}\\
 & \le4\sqrt{d}\frac{\frac{K}{\epsilon}\sigma_{1}\Upsilon}{\frac{\sigma^{2}_{1}}{2}}=\frac{8K\sqrt{d}\sigma_{1}\Upsilon}{\epsilon\sigma^{2}_{1}}.
\end{align*}
By Lemma \ref{exponential-bound}-\ref{lem:initial}
\begin{align*}
\left\Vert y^{(T)}_{-1}\right\Vert _{2} & \le\left\Vert y^{(T)}_{\Head}\right\Vert _{2}+\left\Vert y^{(T)}_{\Tail}\right\Vert _{2}\\
 & \le O\left(\left(\frac{\sqrt{\min\left\{ 4n/\sigma^{2}_{1},d\right\} }}{\epsilon\sigma^{2}_{1}\sqrt{\kappa}}+\frac{1}{\epsilon\sigma^{2}_{1}\kappa}+\frac{\sqrt{d}}{\epsilon\sigma^{2}_{1}}\right)K\sigma_{1}\Upsilon\right)m^{(T)}+O\left(KT\alpha^{T}_{2}\sqrt{d}\right).
\end{align*}
\end{proof}

Having the necessary bounds for the growth rate of each component
of the solution, we now establish the condition and the convergence
of the algorithm. 

\begin{lem}
\label{lem:utility-main}Let $v_{1}$ be the first singular vector
of $A$. Assuming that $\kappa=\frac{\sigma^{2}_{1}-\sigma^{2}_{2}}{\sigma^{2}_{1}}\ge4\left(\frac{K}{\epsilon\sigma_{1}}\Upsilon+\frac{K^{2}}{\epsilon\sigma^{2}_{1}}\right)$,
then, with probability at least $\frac{3}{4}-6\beta$ we have

\[
\sin^{2}\left(x^{(T)},v_{1}\right)=\frac{\left\Vert y^{(T)}_{-1}\right\Vert ^{2}}{\left\Vert y^{(T)}\right\Vert ^{2}_{2}}\le\left(R+O\left(KT\sqrt{d}\right)\left(1+\frac{\kappa}{2}\right)^{-T}\right)^{2}
\]
and
\begin{align*}
\frac{\left\Vert A^{\top}Ax^{(T)}\right\Vert _{2}}{\left\Vert x^{(T)}\right\Vert _{2}} & \ge\frac{\sigma^{2}_{1}}{1+\frac{1}{2}\left(R+O\left(KT\sqrt{d}\right)\left(1+\frac{\kappa}{2}\right)^{-T}\right)^{2}},
\end{align*}
where $R=O\left(\left(\frac{\sqrt{\min\left\{ 4n/\sigma^{2}_{1},d\right\} }}{\epsilon\sigma^{2}_{1}\sqrt{\kappa}}+\frac{1}{\epsilon\sigma^{2}_{1}\kappa}+\frac{\sqrt{d}}{\epsilon\sigma^{2}_{1}}\right)K\sigma_{1}\Upsilon\right)$.
\end{lem}
\begin{proof}
Intersect the initialization event of Lemma \ref{lem:initial} with
the common event used in Lemma \ref{lem:bound-y}, which includes the
recurrence event of Lemma \ref{exponential-bound}. The intersection
has probability at least $3/4-6\beta$, and on it
$m^{(T)}\ge\alpha_1^T/1000$. Therefore, 
\begin{align*}
\left\Vert y^{(T)}\right\Vert ^{2}_{2} & =\left|y^{(T)}_{1}\right|^{2}+\left\Vert y^{(T)}_{-1}\right\Vert ^{2}_{2}\\
 & \le\left(m^{(T)}\right)^{2}+\left(Rm^{(T)}+O\left(KT\alpha^{T}_{2}\sqrt{d}\right)\right)^{2}\\
 & =\left(1+\left(R+\alpha^{T}_{2}\cdot O\left(KT\sqrt{d}\right)\frac{1}{m^{(T)}}\right)^{2}\right)\left(m^{(T)}\right)^{2}\\
 & \le\left(1+\left(R+O\left(KT\sqrt{d}\right)\left(\frac{\alpha_{2}}{\alpha_{1}}\right)^{T}\right)^{2}\right)\left(m^{(T)}\right)^{2}.
\end{align*}
Thus,
\begin{align*}
\frac{\left|y^{(T)}_{1}\right|^{2}}{\left\Vert y^{(T)}\right\Vert ^{2}_{2}} & \ge\frac{1}{1+\left(R+O\left(KT\sqrt{d}\right)\left(1+\frac{\kappa}{2}\right)^{-T}\right)^{2}}
\end{align*}
We then obtain the following results as consequences:
\[
\sin^{2}\left(x^{(T)},v_{1}\right)=\sin^{2}\left(y^{(T)},e_{1}\right)=\frac{\left\Vert y^{(T)}_{-1}\right\Vert ^{2}}{\left\Vert y^{(T)}\right\Vert ^{2}_{2}}\le\left(R+O\left(KT\sqrt{d}\right)\left(1+\frac{\kappa}{2}\right)^{-T}\right)^{2}
\]
and
\[
\frac{\left|y^{(T)}_{1}\right|}{\left\Vert y^{(T)}\right\Vert _{2}}\ge\frac{1}{1+\frac{1}{2}\left(R+O\left(KT\sqrt{d}\right)\left(1+\frac{\kappa}{2}\right)^{-T}\right)^{2}}
\]
and
\begin{align*}
\frac{\left\Vert A^{\top}Ax^{(T)}\right\Vert _{2}}{\left\Vert x^{(T)}\right\Vert _{2}} & =\frac{\left\Vert \Sigma^{2}y^{(T)}\right\Vert _{2}}{\left\Vert y^{(T)}\right\Vert _{2}}\ge\frac{\sigma^{2}_{1}m^{(T)}}{\left(1+\frac{1}{2}\left(R+O\left(KT\sqrt{d}\right)\left(1+\frac{\kappa}{2}\right)^{-T}\right)^{2}\right)m^{(T)}}\\
 & \ge\frac{\sigma^{2}_{1}}{1+\frac{1}{2}\left(R+O\left(KT\sqrt{d}\right)\left(1+\frac{\kappa}{2}\right)^{-T}\right)^{2}}.
\end{align*}

\end{proof}

We are now ready to complete the proof of Theorem \ref{thm:main}.

\begin{proof}
[Proof of Theorem \ref{thm:main}] {Apply Lemma \ref{lem:utility-main} with failure parameter $\beta/6$, giving success probability $3/4-\beta$.} The guarantee of Lemma \ref{lem:utility-main}
is in terms of the per-iteration privacy budget $\epsilon,\delta$.
To complete the proof, it only remains to state the guarantee in terms
of the overall privacy budget $\epstot=2\epsilon\sqrt{2\ln\left(\frac{1}{\delta}\right)T}+2\epsilon\frac{e^{2\epsilon}-1}{e^{2\epsilon}+1}T$
and $\deltot=\left(T+1\right)\delta$. Since $\epstot\leq1$, the
first term in $\epstot$ is larger than the second and thus
\[
\epstot\leq4\epsilon\sqrt{2\ln\left(\frac{1}{\delta}\right)T}=4\epsilon\sqrt{2\ln\left(\frac{T+1}{\deltot}\right)T}
\]
and thus
\[
\epsilon\geq\frac{\epstot}{4\sqrt{2\ln\left(\frac{T+1}{\deltot}\right)T}}.
\]
Plugging this value into the guarantee in Lemma \ref{lem:utility-main}
completes the proof.
\end{proof}

\section{Analysis of Algorithm \ref{alg:Algorithm} for I.I.D. Data \label{subsec:Analysis-gauss}}

\subsection{Analysis of the error $\sin^{2}\left(\overline{v}_{1},v_{1}\right)$}

We will use a similar approach to Theorem 1.1 in \cite{DBLP:conf/colt/JainJKNS16}.
More precisely, we will use the matrix Bernstein inequality to show
that $A^{\top}A$ is close to $n\overline{\Sigma}^{2}$ in spectral
norm, and combine this guarantee with Wedin's theorem in order to
obtain an upper bound on the error $\sin^{2}\left(\overline{v}_{1},v_{1}\right)$.

We start with the lemma.
\begin{lem}
\label{lem:matrix-bernstein-bounds}Let $\beta>0$. Let $g\sim\N(0,\overline{\Sigma}^{2})$
. We have
\[
\left\Vert gg^{\top}-\overline{\Sigma}^{2}\right\Vert _{2}\leq C\left(\sum^{d}_{i=1}\overline{\sigma}^{2}_{i}\right)\log\frac{1}{\beta}\text{ with probability }1-\beta,
\]
and 
\[
\left\Vert \E\left[\left(gg^{\top}-\overline{\Sigma}^{2}\right)^{2}\right]\right\Vert _{2}\leq\overline{\sigma}^{2}_{1}\left(\sum^{d}_{i=1}\overline{\sigma}^{2}_{i}\right)+\overline{\sigma}^{4}_{1}\le2\overline{\sigma}^{2}_{1}\left(\sum^{d}_{i=1}\overline{\sigma}^{2}_{i}\right),
\]
 where $C$ is an absolute constant.
\end{lem}
\begin{proof}

We have
\[
\left\Vert gg^{\top}-\overline{\Sigma}^{2}\right\Vert _{2}\le\left\Vert gg^{\top}\right\Vert _{2}+\left\Vert \overline{\Sigma}^{2}\right\Vert _{2}=\max_{x\colon\left\Vert x\right\Vert _{2}\le1}\left\Vert gg^{\top}x\right\Vert _{2}+\overline{\sigma}^{2}_{1}=\max_{x\colon\left\Vert x\right\Vert _{2}\le1}\left\Vert g\right\Vert _{2}\cdot g^{\top}x+\overline{\sigma}^{2}_{1}=\left\Vert g\right\Vert ^{2}_{2}+\overline{\sigma}^{2}_{1}
\]
Let $t=\ln\frac{1}{\beta}$. Lemma \ref{lem:chi-square} gives that,
with probability at least $1-\exp(-t)=1-\beta$, we have
\begin{align*}
\left\Vert g\right\Vert ^{2}_{2} & \leq\sum_{i}\overline{\sigma}^{2}_{i}+2\sqrt{\sum_{i}\overline{\sigma}^{4}_{i}}\sqrt{t}+2\overline{\sigma}^{2}_{1}t\leq\sum_{i}\overline{\sigma}^{2}_{i}+2\left(\sum_{i}\overline{\sigma}^{2}_{i}\right)\sqrt{t}+2\overline{\sigma}^{2}_{1}t
\end{align*}
and thus
\[
\left\Vert gg^{\top}-\overline{\Sigma}^{2}\right\Vert _{2}\leq\left(1+2\sqrt{t}\right)\left(\sum_{i}\overline{\sigma}^{2}_{i}\right)+\left(2t+1\right)\overline{\sigma}^{2}_{1}\leq\left(2+2\sqrt{t}+2t\right)\left(\sum_{i}\overline{\sigma}^{2}_{i}\right)=O\left(t\right)\left(\sum_{i}\overline{\sigma}^{2}_{i}\right)
\]
as needed.

{For the second inequality, work in an orthonormal eigenbasis of $\overline{\Sigma}^{2}$, so the coordinates of $g$ are independent. Orthogonal changes of basis preserve the matrix norms below.} Since $\E\left[\left(gg^{\top}-\overline{\Sigma}^{2}\right)^{2}\right]$
is a psd matrix, its largest singular value is equal to its largest
eigenvalue. Thus, by the Courant-Fischer theorem, we have
\[
\left\Vert \E\left[\left(gg^{\top}-\overline{\Sigma}^{2}\right)^{2}\right]\right\Vert _{2}=\max_{x\colon\left\Vert x\right\Vert _{2}\leq1}x^{\top}\E\left[\left(gg^{\top}-\overline{\Sigma}^{2}\right)^{2}\right]x
\]
Consider any vector $x\in\R^{d}$. We have
\begin{align*}
x^{\top}\E\left[\left(gg^{\top}-\overline{\Sigma}^{2}\right)^{2}\right]x & =\E\left[x^{\top}\left(gg^{\top}gg^{\top}-gg^{\top}\overline{\Sigma}^{2}-\overline{\Sigma}^{2}gg^{\top}+\overline{\Sigma}^{4}\right)x\right]\\
 & =\E\left[\left\Vert g\right\Vert ^{2}_{2}\left\langle g,x\right\rangle ^{2}\right]-2\E\left[\left\langle g,x\right\rangle \left\langle g,\overline{\Sigma}^{2}x\right\rangle \right]+x^{\top}\overline{\Sigma}^{4}x
\end{align*}
We have
\[
\E\left[\left\Vert g\right\Vert ^{2}_{2}\left\langle g,x\right\rangle ^{2}\right]=\sum_{i,j,k}\E\left[g^{2}_{k}g_{i}g_{j}\right]x_{i}x_{j}
\]
Note that, for any two indices $i\neq j$, using that different coordinates
of $g$ are independent, we obtain
\begin{align*}
\sum_{k}\E\left[g^{2}_{k}g_{i}g_{j}\right] & =\sum_{k=i}\E\left[g^{2}_{k}g_{i}g_{j}\right]+\sum_{k=j}\E\left[g^{2}_{k}g_{i}g_{j}\right]+\sum_{k\neq i,j}\E\left[g^{2}_{k}g_{i}g_{j}\right]\\
 & =\sum_{k=i}\E\left[g^{3}_{i}\right]\underbrace{\E\left[g_{j}\right]}_{=0}+\sum_{k=j}\underbrace{\E\left[g_{i}\right]}_{=0}\E\left[g^{3}_{j}\right]+\sum_{k\neq i,j}\E\left[g^{2}_{k}g_{i}g_{j}\right]\underbrace{\E\left[g_{i}\right]}_{=0}\underbrace{\E\left[g_{j}\right]}_{=0}\\
 & =0
\end{align*}
Thus
\begin{align*}
\E\left[\left\Vert g\right\Vert ^{2}_{2}\left\langle g,x\right\rangle ^{2}\right] & =\sum_{i,k}\E\left[g^{2}_{k}g^{2}_{i}\right]x^{2}_{i}\\
 & =\sum_{i}\left(\E\left[g^{4}_{i}\right]x^{2}_{i}+\sum_{k\neq i}\underbrace{\E\left[g^{2}_{k}\right]}_{=\overline{\sigma}^{2}_{k}}\underbrace{\E\left[g^{2}_{i}\right]}_{=\overline{\sigma}^{2}_{i}}x^{2}_{i}\right)\\
 & \leq\left(\sum_{i}\overline{\sigma}^{2}_{i}\right)\left(\sum_{i}\overline{\sigma}^{2}_{i}x^{2}_{i}\right)+2\sum_{i}\overline{\sigma}^{4}_{i}x^{2}_{i}
\end{align*}
where we have used that $\E\left[g^{2}_{i}\right]=\mathrm{Var}\left[g_{i}\right]=\overline{\sigma}^{2}_{i}$
and $\E\left[g^{4}_{i}\right]={3\overline{\sigma}^{4}_{i}}$.

Since for $i\neq j$, we have $\E\left[g_{i}g_{j}\right]=\E\left[g_{i}\right]\E\left[g_{j}\right]=0$,
\[
\E\left[\left\langle g,x\right\rangle \left\langle g,\overline{\Sigma}^{2}x\right\rangle \right]=\sum_{i,j}\E\left[g_{i}g_{j}\right]x_{i}\overline{\sigma}^{2}_{j}x_{j}=\sum_{i}\E\left[g^{2}_{i}\right]\overline{\sigma}^{2}_{i}x^{2}_{i}=\sum_{i}\overline{\sigma}^{4}_{i}x^{2}_{i}
\]
Therefore
\begin{align*}
x^{\top}\E\left[\left(gg^{\top}-\overline{\Sigma}^{2}\right)^{2}\right]x & =\E\left[\left\Vert g\right\Vert ^{2}_{2}\left\langle g,x\right\rangle ^{2}\right]-2\E\left[\left\langle g,x\right\rangle \left\langle g,\overline{\Sigma}^{2}x\right\rangle \right]+x^{\top}\overline{\Sigma}^{4}x\\
 & \leq\left(\sum_{i}\overline{\sigma}^{2}_{i}\right)\left(\sum_{i}\overline{\sigma}^{2}_{i}x^{2}_{i}\right)+2\sum_{i}\overline{\sigma}^{4}_{i}x^{2}_{i}-2\sum_{i}\overline{\sigma}^{4}_{i}x^{2}_{i}+\sum_{i}\overline{\sigma}^{4}_{i}x^{2}_{i}\\
 & \leq\left(\sum_{i}\overline{\sigma}^{2}_{i}\right)\left(\sum_{i}\overline{\sigma}^{2}_{i}x^{2}_{i}\right)+\sum_{i}\overline{\sigma}^{4}_{i}x^{2}_{i}\\
 & \leq\left(\left(\sum_{i}\overline{\sigma}^{2}_{i}\right){\overline{\sigma}^{2}_{1}}+\overline{\sigma}^{4}_{1}\right)\left(\sum_{i}x^{2}_{i}\right)
\end{align*}
Thus
\[
\left\Vert \E\left[\left(gg^{\top}-\overline{\Sigma}^{2}\right)^{2}\right]\right\Vert _{2}=\max_{x\colon\left\Vert x\right\Vert _{2}\leq1}x^{\top}\E\left[\left(gg^{\top}-\overline{\Sigma}^{2}\right)^{2}\right]x\leq{\overline{\sigma}^{2}_{1}}\left(\sum_{i}\overline{\sigma}^{2}_{i}\right)+\overline{\sigma}^{4}_{1}
\]
as needed.
\end{proof}

We now recall the matrix Bernstein inequality \cite{DBLP:journals/focm/Tropp12}.
\begin{thm}[Matrix Bernstein inequality, Theorem 1.4]
\label{thm:matrix-bernstein}  Let $X_{1},X_{2},\dots,X_{n}\in\R^{d\times d}$
be independent, random, self-adjoint matrices. Assume that each random
matrix satisfies
\[
\E\left[X_{k}\right]=0\text{ and }\left\Vert X_{k}\right\Vert _{2}\leq R\text{ almost surely}.
\]
Then, for all $t\geq0$,
\[
\Pr\left[\left\Vert \sum^{n}_{k=1}X_{k}\right\Vert _{2}\geq t\right]\leq d\cdot\exp\left(\frac{-t^{2}/2}{M+Rt/3}\right)\text{ where }M=\left\Vert \sum^{n}_{k=1}\E\left[X^{2}_{k}\right]\right\Vert _{2}.
\]
\end{thm}
Using the matrix Bernstein inequality together with the guarantee
in Lemma \ref{lem:matrix-bernstein-bounds}, we obtain the following
lemma.
\begin{lem}
\label{lem:bernstein-error}Let \textbf{$\beta>0$}. With probability
at least $1-\beta$, we have
\[
\left\Vert A^{\top}A-n\overline{\Sigma}^{2}\right\Vert _{2}\leq G
\]
where $G=\max\left\{ O\left(\sqrt{n\cdot\overline{\sigma}^{2}_{1}\left(\sum^{d}_{i=1}\overline{\sigma}^{2}_{i}\right)\log\frac{d}{\beta}}\right),O\left(\left(\sum^{d}_{i=1}\overline{\sigma}^{2}_{i}\right)\log\frac{n}{\beta}\log\frac{d}{\beta}\right)\right\} $.
\end{lem}
\begin{proof}
Recall that the rows of $A$ are i.i.d. samples $g_{1},g_{2},\dots,g_{n}$
from the Gaussian distribution $\mathcal{N}(0,\overline{\Sigma}^{2})$.
Thus we have $A^{\top}A=\sum^{n}_{i=1}g_{i}g^{\top}_{i}$.

Let $C$ be an absolute constant as in Lemma \ref{lem:matrix-bernstein-bounds}.
Let $R=C\left(\sum^{d}_{i=1}\overline{\sigma}^{2}_{i}\right)\log\frac{4n^{2}}{\beta}$.
Consider the following matrices $\left\{ X_{i}\right\} _{1\leq i\leq n}$:
\[
X_{i}=\begin{cases}
g_{i}g^{\top}_{i}-\overline{\Sigma}^{2} & \text{if }\left\Vert g_{i}g^{\top}_{i}-\overline{\Sigma}^{2}\right\Vert _{2}\leq R\\
0 & \text{otherwise.}
\end{cases}
\]
Let 
\begin{align*}
Y_{i} & =X_{i}-\E\left[X_{i}\right],
\end{align*}
then $\E\left[Y_{i}\right]=0$ and $\left\Vert Y_{i}\right\Vert _{2}\le2R$.
We have by Cauchy-Schwartz
\begin{align*}
\left\Vert \E\left[X_{i}\right]\right\Vert _{2} & \le\sqrt{\E\left[\left\Vert g_{i}g^{\top}_{i}-\overline{\Sigma}^{2}\right\Vert ^{2}_{2}\right]}\sqrt{\Pr\left[\left\Vert g_{i}g^{\top}_{i}-\overline{\Sigma}^{2}\right\Vert _{2}>R\right]}\\
 & \le\sqrt{6}\sum^{d}_{i=1}\overline{\sigma}^{2}_{i}\sqrt{\Pr\left[\left\Vert g_{i}g^{\top}_{i}-\overline{\Sigma}^{2}\right\Vert _{2}>R\right]}\\
\E\left[Y^{2}_{i}\right] & =\E\left[X^{2}_{i}\right]-\E\left[X_{i}\right]^{2}\preceq\E\left[X^{2}_{i}\right]\preceq\E\left[\left(g_{i}g^{\top}_{i}-\overline{\Sigma}^{2}\right)^{2}\right].
\end{align*}
Let us now show that Lemma \ref{lem:matrix-bernstein-bounds} implies
the following:
\begin{equation}
\sum^{n}_{i=1}X_{i}=A^{\top}A-n\overline{\Sigma}^{2}\text{ with probability }1-\frac{\beta}{2}\label{eq:b1}
\end{equation}
and
\begin{equation}
\left\Vert \sum^{n}_{i=1}\E\left[Y^{2}_{i}\right]\right\Vert _{2}\leq M=n\cdot\overline{\sigma}^{2}_{1}\left(\sum^{d}_{i=1}\overline{\sigma}^{2}_{i}\right)\label{eq:b2}
\end{equation}
To show (\ref{eq:b1}), we apply Lemma \ref{lem:matrix-bernstein-bounds}
with probability parameter $\frac{\beta}{4n^{2}}$, and obtain that
$\Pr\left[\left\Vert g_{i}g^{\top}_{i}-\overline{\Sigma}^{2}\right\Vert _{2}\geq R\right]\leq\frac{\beta}{4n^{2}}$.
By taking an union bound over these $n$ evens, we obtain that $\Pr\left[\exists i\colon\left\Vert g_{i}g^{\top}_{i}-\overline{\Sigma}^{2}\right\Vert _{2}\geq R\right]\leq\frac{\beta}{2}$.
Thus, with probability at least $1-\frac{\beta}{2}$, we have $\sum^{n}_{i=1}X_{i}=\sum^{n}_{i=1}\left(g_{i}g^{\top}_{i}-\overline{\Sigma}^{2}\right)=A^{\top}A-n\overline{\Sigma}^{2}$.

To show (\ref{eq:b2}), we note that, since the $g_{i}$ are i.i.d.
samples from $\mathcal{N}(0,\overline{\Sigma}^{2})$, we have
\[
\left\Vert \sum^{n}_{i=1}\E\left[\left(g_{i}g^{\top}_{i}-\overline{\Sigma}^{2}\right)^{2}\right]\right\Vert _{2}=n\cdot\left\Vert \E_{g\sim\mathcal{N}(0,\overline{\Sigma}^{2})}\left[\left(gg^{\top}-\overline{\Sigma}^{2}\right)^{2}\right]\right\Vert _{2}\leq2n\cdot\overline{\sigma}^{2}_{1}\left(\sum^{d}_{i=1}\overline{\sigma}^{2}_{i}\right),
\]
where the inequality follows from Lemma \ref{lem:matrix-bernstein-bounds}.

Note that the matrices $Y_{i}$ are symmetric. Additionally, we have
$\left\Vert Y_{i}\right\Vert _{2}\leq2R$ by definition and $\left\Vert \sum^{n}_{i=1}\E\left[Y^{2}_{i}\right]\right\Vert _{2}\leq M$
by (\ref{eq:b2}). Thus, by the matrix Bernstein inequality (Theorem
\ref{thm:matrix-bernstein}), for all $t\geq0$,
\[
\Pr\left[\left\Vert \sum^{n}_{i=1}Y_{i}\right\Vert _{2}\geq t\right]\le d\cdot\exp\left(\frac{-t^{2}/2}{M+2Rt/3}\right).
\]
For $t=G:=\max\left\{ 2\sqrt{2M\ln\frac{2d}{\beta}},12R\ln\frac{2d}{\beta}\right\} $,
the right-hand side of the inequality above is at most $\frac{\beta}{2}$.
By combining with (\ref{eq:b1}), we obtain that, with probability
at least $1-\beta$, we have
\begin{align*}
\left\Vert A^{\top}A-n\overline{\Sigma}^{2}\right\Vert _{2} & =\left\Vert \sum^{n}_{i=1}X_{i}\right\Vert _{2}=\left\Vert \sum^{n}_{i=1}Y_{i}+\sum^{n}_{i=1}\E\left[X_{i}\right]\right\Vert _{2}\\
 & \le\left\Vert \sum^{n}_{i=1}Y_{i}\right\Vert _{2}+n\left\Vert \E\left[X_{1}\right]\right\Vert _{2}\\
 & \leq G+n\sqrt{6}\sum^{d}_{i=1}\overline{\sigma}^{2}_{i}.\sqrt{\frac{\beta}{4n^{2}}}\\
 & =\max\left\{ O\left(\sqrt{n\cdot\overline{\sigma}^{2}_{1}\left(\sum^{d}_{i=1}\overline{\sigma}^{2}_{i}\right)\log\frac{d}{\beta}}\right),O\left(\left(\sum^{d}_{i=1}\overline{\sigma}^{2}_{i}\right)\log\frac{n}{\beta}\log\frac{d}{\beta}\right)\right\} .
\end{align*}
\end{proof}

Next, we recall Wedin's theorem.
\begin{thm}[Wedin's theorem, \cite{wedin1972perturbation}]
\label{thm:wedin} Let $\widehat{\Sigma}=\Sigma+E$. Let $\sigma_{i}$
be the $i$-th largest singular value of $\Sigma$. Let $v_{1}$ and
$\widehat{v}_{1}$ be the largest singular vectors of $\Sigma$ and
$\widehat{\Sigma}$, respectively. We have
\[
\sin\left(v_{1},\widehat{v}_{1}\right)\leq C\cdot\frac{\left\Vert E\right\Vert _{2}}{\sigma_{1}-\sigma_{2}},
\]
where $C$ is an absolute constant.
\end{thm}
Finally, by combining Wedin's theorem with Lemma \ref{lem:bernstein-error},
we obtain the following guarantee on the error $\sin^{2}\left(\overline{v}_{1},v_{1}\right)$.
\begin{thm}
\label{thm:gaussian-error1}Let \textbf{$\beta>0$}. With probability
at least $1-\beta$, we have
\[
\sin^{2}\left(\overline{v}_{1},v_{1}\right)\leq\widetilde{O}\left(\frac{\sum^{d}_{i=1}\overline{\sigma}^{2}_{i}}{n\overline{\sigma}^{2}_{1}\overline{\kappa}^{2}}+\left(\frac{\sum^{d}_{i=1}\overline{\sigma}^{2}_{i}}{n\overline{\sigma}^{2}_{1}\overline{\kappa}}\right)^{2}\right).
\]
\end{thm}
\begin{proof}
By Lemma \ref{lem:bernstein-error}, with probability at least $1-\beta$,
we have
\begin{align*}
\left\Vert A^{\top}A-n\overline{\Sigma}^{2}\right\Vert _{2} & \leq K\cdot\max\left\{ \sqrt{n\cdot\overline{\sigma}^{2}_{1}\left(\sum^{d}_{i=1}\overline{\sigma}^{2}_{i}\right)},\sum^{d}_{i=1}\overline{\sigma}^{2}_{i}\right\} 
\end{align*}
where $K=O\left(\max\left\{ \sqrt{\log\frac{d}{\beta}},\log\frac{n}{\beta}\log\frac{d}{\beta}\right\} \right)=\widetilde{O}\left(1\right)$.

In the following, we assume the above inequality holds.

Next, we apply Wedin's theorem (Theorem \ref{thm:wedin}) with $\Sigma=n\overline{\Sigma}^{2}$
and $\widehat{\Sigma}=A^{\top}A$. Since $\sigma_{1}(\Sigma)=n\overline{\sigma}^{2}_{1}$
and $\sigma_{2}(\Sigma)=n\overline{\sigma}^{2}_{2}$, we obtain
\begin{align*}
\sin\left(\overline{v}_{1},v_{1}\right) & \leq C\cdot\frac{\left\Vert A^{\top}A-n\overline{\Sigma}^{2}\right\Vert _{2}}{n\left(\overline{\sigma}^{2}_{1}-\overline{\sigma}^{2}_{2}\right)}\\
 & \leq CK\frac{\max\left\{ \sqrt{n\cdot\overline{\sigma}^{2}_{1}\left(\sum^{d}_{i=1}\overline{\sigma}^{2}_{i}\right)},\sum^{d}_{i=1}\overline{\sigma}^{2}_{i}\right\} }{n\left(\overline{\sigma}^{2}_{1}-\overline{\sigma}^{2}_{2}\right)}\\
 & =CK\frac{\max\left\{ \sqrt{n\cdot\overline{\sigma}^{2}_{1}\left(\sum^{d}_{i=1}\overline{\sigma}^{2}_{i}\right)},\sum^{d}_{i=1}\overline{\sigma}^{2}_{i}\right\} }{n\overline{\sigma}^{2}_{1}\overline{\kappa}}
\end{align*}
where $C$ is an absolute constant. Thus,
\begin{align*}
\sin^{2}\left(\overline{v}_{1},v_{1}\right) & \leq\widetilde{O}\left(\frac{\sum^{d}_{i=1}\overline{\sigma}^{2}_{i}}{n\overline{\sigma}^{2}_{1}\overline{\kappa}^{2}}+\left(\frac{\sum^{d}_{i=1}\overline{\sigma}^{2}_{i}}{n\overline{\sigma}^{2}_{1}\overline{\kappa}}\right)^{2}\right)
\end{align*}
as needed.
\end{proof}

\subsection{Analysis of the error $\sin^{2}\left(x,v_{1}\right)$}

To upper bound the error $\sin^{2}\left(x,v_{1}\right)$, we will
use our main result shown in Corollary \ref{cor:main}. To this end,
we need to show that the conditions of Corollary \ref{cor:main} are
satisfied. Additionally, since the guarantee of Corollary \ref{cor:main}
is in terms of the values $\Upsilon$, $\kappa$, $\sigma_{1}$ of
the matrix $A^{\top}A$, we will need to find suitable bounds for
these quantities.

We start by deriving a guarantee on the length of each row of $A$,
which follows immediately from Lemma \ref{lem:chi-square} (Lemma
1 in \cite{laurent2000adaptive}).
\begin{lem}
\label{lem:gaussian-row-length-A} Let $\beta>0$. With probability
at least $1-\beta$, every row $a_{i}$ of $A$ has length at most
$C\sqrt{\log\frac{n}{\beta}}\sqrt{\sum_{i}\overline{\sigma}^{2}_{i}}$
where $C$ is an absolute constant.
\end{lem}
\begin{proof}
Let $t=\ln\frac{n}{\beta}$. Consider a row $a_{i}$ of $A$. Since
$a_{i}\sim\mathcal{N}(0,\overline{\Sigma}^{2})$, Lemma \ref{lem:chi-square}
gives that, with probability at least $1-\exp(-t)=1-\frac{\beta}{n}$,
we have
\begin{align*}
\left\Vert a_{i}\right\Vert _{2} & \leq\left(\sum_{i}\overline{\sigma}^{2}_{i}+2\sqrt{\sum_{i}\overline{\sigma}^{4}_{i}}\sqrt{t}+2\overline{\sigma}^{2}_{1}t\right)^{1/2}\\
 & \leq\left(\sum_{i}\overline{\sigma}^{2}_{i}+2\left(\sum_{i}\overline{\sigma}^{2}_{i}\right)\sqrt{t}+2\left(\sum_{i}\overline{\sigma}^{2}_{i}\right)t\right)^{1/2}\\
 & =\sqrt{1+2\sqrt{t}+2t}\sqrt{\sum_{i}\overline{\sigma}^{2}_{i}}\\
 & \leq\left(1+\sqrt{2t}\right)\sqrt{\sum_{i}\overline{\sigma}^{2}_{i}}\\
 & :=K
\end{align*}
The claim now follows by taking the union bound over the $n$ events
$\left\Vert a_{i}\right\Vert _{2}\geq K$.
\end{proof}

\begin{lem}
\label{lem:gaussian-coherence}Let ${\beta\in(0,1)}$.
{Assume $n\ge d$ and $\overline{\sigma}_1>0$.} With probability
at least $1-\beta$, we have $\Upsilon\le {O\left(\frac{\sqrt{\log\frac{2n}{\beta}}}{\sqrt{n}}\right)}$.
\end{lem}
\begin{proof}
For every orthogonal $O\in\mathbb R^{n\times n}$, $OA$ and $A$
have the same distribution and the line spanned by $U_{:,1}$ is rotationally invariant.
After independently randomizing its sign, $U_{:,1}$ is uniform on
$S^{n-1}$ and has the same distribution as $x/\|x\|_2$ for
$x\sim\mathcal N(0,I_n)$. This sign randomization does not change
$\Upsilon$.

Gaussian tails and a union bound give
\[
\Pr\left[\max_i|x_i|>\sqrt{2\log(4n/\beta)}\right]\le\beta/2.
\]
The lower-tail Gaussian quadratic-form bound also gives us
\[
\Pr\left[\|x\|_2^2<n-2\sqrt{n\log(2/\beta)}\right]\le\beta/2.
\]
If $n\ge16\log(2/\beta)$, the latter lower bound is at least $n/2$.
Consequently, with probability at least $1-\beta$,
\[
\Upsilon\le2\sqrt{\frac{\log(4n/\beta)}{n}}.
\]
If $n<16\log(2/\beta)$, use $\Upsilon\le1$ and
$4\sqrt{\log(4n/\beta)/n}>1$. This proves the stated order in all cases.
\end{proof}

We recall Weyl's inequality.
\begin{thm}[Weyl's inequality]
\label{thm:weyl}  Let $\widehat{\Sigma}=\Sigma+E$, where $\Sigma$
and $E$ are symmetric matrices in $\R^{d\times d}$. Let $\lambda_{i}$
and $\widehat{\lambda}_{i}$ be the $i$-th largest eigenvalues of
$\Sigma$ and $\widehat{\Sigma}$, respectively. For each $1\leq i\leq d$,
we have
\[
\left|\lambda_{i}-\widehat{\lambda}_{i}\right|\leq\left\Vert E\right\Vert _{2}.
\]
\end{thm}
Using Weyl's inequality, we show the following lemma that relates
the eigengap and first singular value of $A^{\top}A$ to the corresponding
values of $\overline{\Sigma}^{2}$.
\begin{lem}
\label{lem:gaussian-cor-condition}Suppose that $\sum^{d}_{i=1}\overline{\sigma}^{2}_{i}=1$.
{Assume $\overline{\sigma}_1>0$ and $\overline{\kappa}>0$.}
Suppose that $\left\Vert A^{\top}A-n\overline{\Sigma}^{2}\right\Vert _{2}\leq K_{3}\cdot\max\left\{ \sqrt{n}\overline{\sigma}_{1},1\right\} $.
If $n\geq\max\left\{ d,{\frac{(4K_3)^2}{\overline{\sigma}_1^2\overline{\kappa}^2}}\right\} $,
we have ${\frac13\overline{\kappa}\le\kappa\le2\overline{\kappa}}\text{ and \ensuremath{\sigma}}^{2}_{1}\geq\frac{1}{2}n\overline{\sigma}^{2}_{1}$.
\end{lem}
\begin{proof}
We apply Weyl's theorem (Theorem \ref{thm:weyl}) with $\Sigma=n\overline{\Sigma}^{2}$
and $\widehat{\Sigma}=A^{\top}A$. Since $\lambda_{i}(\Sigma)=n\overline{\sigma}^{2}_{i}$
and $\lambda_{i}(\widehat{\Sigma})=\sigma^{2}_{i}$, we obtain
\begin{align*}
\left|\sigma^{2}_{i}-n\overline{\sigma}^{2}_{i}\right| & \leq\left\Vert A^{\top}A-n\overline{\Sigma}^{2}\right\Vert _{2}\leq K_{3}\sqrt{n}\overline{\sigma}_{1}\le\frac{1}{4}n\overline{\sigma}^{2}_{1}\overline{\kappa}.
\end{align*}
Suppose that $n\geq d$. Since $\sum^{d}_{i=1}\overline{\sigma}^{2}_{i}=1$,
we have $\overline{\sigma}_{1}\geq\frac{1}{\sqrt{d}}$ and thus $\max\left\{ \sqrt{n}\overline{\sigma}_{1},1\right\} =\sqrt{n}\overline{\sigma}_{1}$.
Thus we have
\[
\left|\sigma^{2}_{i}-n\overline{\sigma}^{2}_{i}\right|\leq K_{3}\cdot\sqrt{n}\overline{\sigma}_{1}
\]
and thus
\begin{align*}
\sigma^{2}_{1} & \geq n\overline{\sigma}^{2}_{1}-K_{3}\cdot\sqrt{n}\overline{\sigma}_{1}\\
\sigma^{2}_{1} & \leq n\overline{\sigma}^{2}_{1}+K_{3}\cdot\sqrt{n}\overline{\sigma}_{1}\\
\sigma^{2}_{2} & \leq n\overline{\sigma}^{2}_{2}+K_{3}\cdot\sqrt{n}\overline{\sigma}_{1}
\end{align*}
We have
\begin{align*}
K_{3}\sqrt{n}\overline{\sigma}_{1}\le & \frac{1}{4}n\overline{\sigma}^{2}_{1}\overline{\kappa}\leq\frac{1}{2}n\overline{\sigma}^{2}_{1}.
\end{align*}
Thus we have $\frac{1}{2}n\overline{\sigma}^{2}_{1}\leq\sigma^{2}_{1}\leq\frac{3}{2}n\overline{\sigma}^{2}_{1}$
and
\begin{align*}
\kappa & =\frac{\sigma^{2}_{1}-\sigma^{2}_{2}}{\sigma^{2}_{1}}\geq\frac{n\left(\overline{\sigma}^{2}_{1}-\overline{\sigma}^{2}_{2}\right)-2K_{3}\cdot\sqrt{n}\overline{\sigma}_{1}}{\frac{3}{2}n\overline{\sigma}^{2}_{1}}=\frac{n\overline{\sigma}^{2}_{1}\overline{\kappa}-2K_{3}\cdot\sqrt{n}\overline{\sigma}_{1}}{\frac{3}{2}n\overline{\sigma}^{2}_{1}}=\frac{2}{3}\overline{\kappa}-\frac{4K_{3}}{3\sqrt{n}\overline{\sigma}_{1}}
\end{align*}
Suppose that $n\geq{\frac{(4K_3)^2}{\overline{\sigma}_1^2\overline{\kappa}^2}}$.
We have $\frac{4K_{3}}{3\sqrt{n}\overline{\sigma}_{1}}\leq\frac{1}{3}\overline{\kappa}$.
The same Weyl bounds also give
\[
\kappa\le\frac{n\overline{\sigma}_1^2\overline{\kappa}+2K_3\sqrt n\,\overline{\sigma}_1}{n\overline{\sigma}_1^2-K_3\sqrt n\,\overline{\sigma}_1}\le2\overline{\kappa}.
\]
Putting everything together, we have that, if $n\geq\max\left\{ d,{\frac{(4K_3)^2}{\overline{\sigma}_1^2\overline{\kappa}^2}}\right\} $,
we have
\begin{align*}
\kappa & \geq\frac{1}{3}\overline{\kappa}\text{ and \ensuremath{\sigma}}^{2}_{1}\geq\frac{1}{2}n\overline{\sigma}^{2}_{1}.
\end{align*}
\end{proof}

By combining our main Corollary \ref{cor:main} with Lemmas \ref{lem:gaussian-row-length-A},
\ref{lem:bernstein-error}, \ref{lem:gaussian-coherence}, and \ref{lem:gaussian-cor-condition},
we obtain the following guarantee on the error.
\begin{thm}
\label{thm:gaussian-error2}Let ${\epsilon,\delta,\beta\in(0,1)}$. Suppose that $\sum^{d}_{i=1}\overline{\sigma}^{2}_{i}=1$.
{Assume $\overline{\sigma}_1>0$ and $\overline{\kappa}>0$, calibrate the per-iteration budgets to total budget $(\epsilon,\delta)$ as in Theorem \ref{thm:main}, and use failure parameter $\beta/24$.}
Let $L=C\sqrt{\log\frac{4n}{\beta}}$ where $C$ is the absolute constant
in \ref{lem:gaussian-row-length-A}. If $n\geq\widetilde{\Theta}\left(d+\frac{1}{\overline{\sigma}^{2}_{1}\overline{\kappa}^{2}}+\frac{1}{\epsilon\overline{\sigma}^{2}_{1}\overline{\kappa}^{3/2}}\right)$
sufficiently large, with probability at least $\frac{3}{4}-\beta$,
when run on the matrix $\frac{1}{L}A$ for $T=\widetilde{\Theta}\left(\frac{1}{\overline{\kappa}}\right)$
iterations, Algorithm \ref{alg:Algorithm} outputs a vector $x=\frac{x^{(T)}}{\left\Vert x^{(T)}\right\Vert _{2}}$
such that 
\[
\sin^{2}\left(v_{1},x\right)\leq\widetilde{O}\left(\frac{d}{\epsilon^{2}n^{2}\overline{\sigma}^{2}_{1}\overline{\kappa}^{2}}+\frac{1}{\epsilon^{2}n^{2}\overline{\sigma}^{2}_{1}\overline{\kappa}^{3}}\right).
\]
\end{thm}
\begin{proof}
In this proof, let $\epsilon,\delta$ denote the total privacy budget.
Set $\delta_0=\delta/(T+1)$ and choose $\epsilon_0$ by the privacy composition
theorem in Theorem \ref{thm:main}. Then
\[
\epsilon_0=\Theta\left(
\frac{\epsilon}{\sqrt{T\log((T+1)/\delta)}}\right).
\]
We use the parameters $J,c_1,c_2,K$  evaluated using
$\epsilon_0,\delta_0$ and failure parameter ${\beta/24}$, as in  Corollary \ref{cor:main}.
The quantities $\sigma_1,\kappa,\Upsilon$ below refer to the unscaled
matrix $A$; scaling to $A/L$ replaces $\sigma_1$ by $\sigma_1/L$
and preserves $\kappa,\Upsilon$ and the singular vectors.
Since $\epsilon\ge2\epsilon_0\sqrt{2T\log(1/\delta_0)}$,
the sufficient gap condition below with coefficient $8/\epsilon_0$
implies the total-budget gap condition in Theorem \ref{thm:main}.

In order to apply Corollary \ref{cor:main}, we need that the following
conditions are satisfied (when applied with probability parameter
$\frac{\beta}{4}$):
\begin{itemize}
\item The row lengths of the input matrix $\frac{1}{L}A$ are at most 1.
\item The number of iterations $T$ is at least $K_{2}\frac{1}{\kappa}$
where $K_{2}=\Theta\left(\log\frac{4n}{\beta\delta\epsilon}\right)$.
\item We have $n\geq d$ and ${\kappa\geq8\left(\frac{KL\Upsilon}{\epsilon_0\sigma_1}+\frac{K^2L^2}{\epsilon_0\sigma_1^2}\right)}$
{where $K=c_1+c_2$.}
\end{itemize}
In the following, we show that the above conditions are satisfied
with probability $1-\frac{3\beta}{4}$ provided $n$ is sufficiently
large. To this end, we apply Lemmas \ref{lem:gaussian-row-length-A},
\ref{lem:bernstein-error}, \ref{lem:gaussian-coherence} (each with
probability parameter $\frac{\beta}{4}$) and obtain that, with probability
at least $1-\frac{3\beta}{4}$ all of the following hold for $L=O\left(\sqrt{\log\frac{4n}{\beta}}\right)$,
$K_{3}=O\left(\max\left\{ \sqrt{\log\frac{4d}{\beta}},\log\frac{4n}{\beta}\log\frac{4d}{\beta}\right\} \right)$,
$K_{4}=O\left(\sqrt{\log\frac{4n}{\beta}}\right)$:
\begin{itemize}
\item Each row of $A$ has length at most $L$, and thus each row of $\frac{1}{L}A$
has length at most $1$.
\item $\left\Vert A^{\top}A-n\overline{\Sigma}^{2}\right\Vert _{2}\leq K_{3}\cdot\max\left\{ \sqrt{n}\overline{\sigma}_{1},1\right\} .$
\item $\Upsilon\leq K_{4}\cdot\frac{1}{\sqrt{n}}$.
\end{itemize}
In the following, we condition on the event that the above guarantees
hold. Note that the first condition required by the corollary is satisfied.
Suppose that
\[
n\geq\max\left\{ d,{\frac{(4K_3)^2}{\overline{\sigma}_1^2\overline{\kappa}^2}}\right\} .
\]
By Lemma \ref{lem:gaussian-cor-condition}, we have ${\frac13\overline{\kappa}\le\kappa\le2\overline{\kappa}}$
and $\sigma^{2}_{1}\geq\frac{1}{2}n\overline{\sigma}^{2}_{1}$. We
set $T=3K_{2}\frac{1}{\overline{\kappa}}$. Since $\kappa\geq\frac{1}{3}\overline{\kappa}$,
the second condition is satisfied. Finally, consider the third condition. Since
$\sigma_1^2\ge n\overline{\sigma}_1^2/2$ and
$\Upsilon\le K_4/\sqrt n$, we have
\begin{align*}
8\left(\frac{KL\Upsilon}{\epsilon_0\sigma_1}
+\frac{K^2L^2}{\epsilon_0\sigma_1^2}\right)
&\le\frac{8}{\epsilon_0 n}
\left(\frac{\sqrt2\,KLK_4}{\overline{\sigma}_1}
+\frac{2K^2L^2}{\overline{\sigma}_1^2}\right).
\end{align*}
This is at most $\overline{\kappa}/3\le\kappa$ provided
\[
n\ge\frac{24}{\epsilon_0\overline{\kappa}}
\left(\frac{\sqrt2\,KLK_4}{\overline{\sigma}_1}
+\frac{2K^2L^2}{\overline{\sigma}_1^2}\right).
\]
Using 
$T=\widetilde{\Theta}(1/\overline{\kappa})$ and
$\overline{\sigma}_1\le1$, this requirement is implied by
\[
n\ge\widetilde{\Omega}\left(
\frac{1}{\epsilon\overline{\sigma}_1^2\overline{\kappa}^{3/2}}
\right).
\]
Together with the condition for Lemma \ref{lem:gaussian-cor-condition},
all the hypotheses hold under the stated sufficient sample size
\[
n\ge\widetilde{\Omega}\left(
d+\frac{1}{\overline{\sigma}_1^2\overline{\kappa}^2}
+\frac{1}{\epsilon\overline{\sigma}_1^2\overline{\kappa}^{3/2}}
\right).
\]
By applying Corollary \ref{cor:main} with probability parameter $\frac{\beta}{4}$,
we obtain that, with probability at least $\frac{3}{4}-\frac{\beta}{4}$
, we have (for $K_{5}=\widetilde{O}(1)$):
\begin{align*}
\sin^{2}\left(v_{1},x\right) &{ \leq K_{5}\cdot\left(\frac{\min\left\{ 4nL^2/\sigma^{2}_{1},d\right\} }{\epsilon^{2}\sigma^{2}_{1}\kappa^{2}}+\frac{1}{\epsilon^{2}\sigma^{2}_{1}\kappa^{3}}+\frac{d}{\epsilon^{2}\sigma^{2}_{1}\kappa}\right)L^2\Upsilon^{2}}\\
 &{\leq K_{5}\cdot\left(\frac{2d}{\epsilon^{2}\sigma^{2}_{1}\kappa^{2}}+\frac{1}{\epsilon^{2}\sigma^{2}_{1}\kappa^{3}}\right)L^2\Upsilon^{2}}\\
 &{ \leq K_{4}^2K_{5}L^2\cdot\left(\frac{4\cdot3^{2}d}{\epsilon^{2}n\overline{\sigma}^{2}_{1}\overline{\kappa}^{2}}+\frac{2\cdot3^{3}}{\epsilon^{2}n\overline{\sigma}^{2}_{1}\overline{\kappa}^{3}}\right)\frac{1}{n}}\\
 &{ =K_{4}^2K_{5}L^2\cdot\left(\frac{4\cdot3^{2}d}{\epsilon^{2}n^{2}\overline{\sigma}^{2}_{1}\overline{\kappa}^{2}}+\frac{2\cdot3^{3}}{\epsilon^{2}n^{2}\overline{\sigma}^{2}_{1}\overline{\kappa}^{3}}\right)}\\
 & =\widetilde{O}\left(\frac{d}{\epsilon^{2}n^{2}\overline{\sigma}^{2}_{1}\overline{\kappa}^{2}}+\frac{1}{\epsilon^{2}n^{2}\overline{\sigma}^{2}_{1}\overline{\kappa}^{3}}\right)
\end{align*}
where we used that $\Upsilon\leq K_{4}\frac{1}{\sqrt{n}}$, $\kappa\geq\frac{1}{3}\overline{\kappa}$,
and $\sigma^{2}_{1}\geq\frac{1}{2}n\overline{\sigma}^{2}_{1}$. This happens with probability $3/4-\beta$.

\end{proof}

\subsection{Proof of Theorem \ref{thm:gaussian}}

Using standard arguments, we can show the following lemma.
\begin{lem}
\label{lem:sin-squared}Let $u,v,w$ be unit length vectors. We have
$\sin^{2}\left(u,w\right)\leq2\left(\sin^{2}\left(u,v\right)+\sin^{2}\left(v,w\right)\right)$.
\end{lem}
Theorem \ref{thm:gaussian} now follows by combining Theorem \ref{thm:gaussian-error1},
Theorem \ref{thm:gaussian-error2}, and Lemma \ref{lem:sin-squared}
{with failure parameter $\beta/2$ in each of the two theorems, giving success probability at least $3/4-\beta$}.

\begin{proof}
For unit vectors $a,b$,
$\|aa^\top-bb^\top\|_2=\sqrt{1-\langle a,b\rangle^2}=\sin(a,b)$.
The triangle inequality therefore gives
\[
\sin(u,w)=\|uu^\top-ww^\top\|_2
\le\|uu^\top-vv^\top\|_2+\|vv^\top-ww^\top\|_2
=\sin(u,v)+\sin(v,w).
\]
Squaring and using $(a+b)^2\le2a^2+2b^2$ proves the claim.
\end{proof}

\end{document}